\documentclass[10pt,journal]{IEEEtran} 
\usepackage{cite}
\usepackage{graphics} % for pdf, bitmapped graphics files
\usepackage{graphicx,color}
\usepackage[active]{srcltx}
\usepackage{amsmath} % assumes amsmath package installed
\usepackage{amssymb} % assumes amsmath package installed
\usepackage{amsmath,amssymb,amsfonts,theorem}
\usepackage{enumerate}
\usepackage{tikz}
\usepackage{tabu}
\usepackage{algpseudocode} 
\usepackage{accents}

\usepackage{epstopdf}

{\theorembodyfont{\normalfont} %\theorembodyfont{\rmfamily}

\newcommand{\R}{\mathbb{R}}

\newcommand{\E}{\mathcal{E}}    %for graphs
\renewcommand{\S}{\mathbb{S}}    %for graphs

\newcommand\ie{\emph{i.e.}}

\newcommand{\sign}{{\rm sign}}

\newtheorem{example}{Example}
\newtheorem{remark}{Remark}
\newtheorem{assumption}{Assumption}
\newtheorem{definition}{Definition}
\newtheorem{theorem}{Theorem}
\newtheorem{property}{Property}
\newtheorem{lemma}{Lemma}
\newtheorem{problem}{Problem}
\newtheorem{corollary}{Corollary}
\newtheorem{proposition}{Proposition}

\newtheorem{manualassumptioninner}{Assumption}
\newenvironment{manualassumption}[1]{%
  \renewcommand\themanualassumptioninner{#1}%
  \manualassumptioninner
}{\endmanualassumptioninner}

\def\QEDopen{{\setlength{\fboxsep}{0pt}\setlength{\fboxrule}{0.2pt}\fbox{\rule[0pt]{0pt}{1.3ex}\rule[0pt]{1.3ex}{0pt}}}}
\def\QED{\QEDopen}

\def\be{\begin{equation}}
\def\ee{\end{equation}}
\def\ba{\begin{array}}
\def\ea{\end{array}}
\def\eqa{\begin{eqnarray}}
\def\eqe{\end{eqnarray}}
\begin{document}

\title{\LARGE \bf Stochastic Stability of Discrete-time Phase-coupled Oscillators\\ over Uncertain and Random Networks}%\\[2mm] 
\author{Matin Jafarian,~\IEEEmembership{Member,~IEEE},
 Mohammad H. Mamduhi,~\IEEEmembership{Senior Member,~IEEE},\\
 and Karl H. Johansson,~\IEEEmembership{Fellow,~IEEE}
\thanks{M. Jafarian is with the Delft Center for Systems and Control, Delft University of Technology, The Netherlands. Email: {\tt\small m.jafarian@tudelft.nl.} Her work was supported by the Horizon 2020 Marie-Curie Fellowship, Project ReWoMen.}
\thanks{M. H. Mamduhi is with the Automatic Control Laboratory, ETH Z\"urich, Switzerland. Email:{\tt\small mmamduhi@ethz.ch.}}
\thanks{K. H. Johansson is with the Division of Decision and Control Systems, School of Electrical Engineering and Computer Science, KTH Royal Institute of Technology, Stockholm, Sweden. He is also affiliated with Digital Futures. Email: {\tt\small kallej@kth.se.} His work is supported by the Knut and Alice Wallenberg Foundation and the Swedish Research Council.}}

\maketitle
%================================================================================================================================================
\begin{abstract}
This article studies stochastic relative phase stability, \ie,\ stochastic phase-cohesiveness, of discrete-time phase-coupled oscillators. Stochastic phase-cohesiveness in two types of networks is studied. First, we consider oscillators coupled with $2\pi$-periodic odd functions over underlying undirected graphs subject to both multiplicative and additive stochastic uncertainties. We prove stochastic phase-cohesiveness of the network with respect to two specific, namely in-phase and anti-phase, sets by deriving sufficient coupling conditions. We show the dependency of these conditions on the size of the mean values of additive and multiplicative uncertainties, as well as the sign of the mean values of multiplicative uncertainties. Furthermore, we discuss the results under a relaxation of the odd property of the coupling function. Second, we study an uncertain network in which the multiplicative uncertainties are governed by the Bernoulli process representing the well-known Erd\"{o}s-R\'{e}nyi network. We assume constant exogenous frequencies and derive sufficient conditions for achieving both stochastic phase-cohesive and phase-locked solutions, \ie, stochastic phase-cohesiveness with respect to the origin. For the latter case, where identical exogenous frequencies are assumed, we prove that any positive probability of connectivity leads to phase-locking. Thorough analyses are provided, and insights obtained from stochastic analysis are discussed, along with numerical simulations to validate the analytical results.
\end{abstract}
%================================================================================================================================================
\section{Introduction}\label{sec:int}
\noindent Oscillatory behavior is a fundamental feature of complex networks in a wide range of applications from mechanical and electrical power networks to biological and neuronal networks \cite{izhikevich2007dynamical,sepulchre2006oscillators,strogatz2000kuramoto,jafarian2019synchronization,nozari2019oscillations,villani2020analysis}. Synchronization is useful for achieving stable oscillations in these networks \cite{sepulchre2006oscillators}. Uncertainties and randomness influence the ability of a network to synchronize. A prominent example is a network of biological neurons whose activities govern the underlying mechanisms of cognition \cite{izhikevich2007dynamical} and motion \cite{ijspeert2013dynamical}. In such neuronal networks, the coupling and the input current of neurons are often subject to fluctuations due to various stochastic uncertainties \cite{borgers2003synchronization,deco2009key,richardson2005synaptic,white2000channel}. Therefore, studying the stability of stochastic uncertain oscillatory models is essential. 

Analyzing the behavior of coupled oscillators is hard even in the absence of uncertainties, due to the nonlinear nature of their dynamics. One approach to reduce the complexity of analyses is to map the oscillators' models to their phase dynamics. Such a transformation is feasible under particular conditions \cite{khalil2002nonlinear,izhikevich2006weakly,sacre2014sensitivity}. The framework has been employed in studying synchronization of some classes of nonlinear coupled oscillators, for instance, van der Pol oscillators, and weakly coupled neurons~\cite{brown2003globally,izhikevich2007dynamical,izhikevich1999weakly}. Among notions of synchronization is phase-cohesiveness, a principal desired behavior of oscillatory networks implying boundedness of relative phases~\cite{dorfler2014synchronization}. Phase-cohesiveness extends the notion of phase-locking (also called phase-synchronization), and is useful in achieving frequency synchronization~\cite{dorfler2014synchronization}. Despite the importance of stochastic synchronization, thus far, phase-cohesiveness in stochastic models has not been studied in the literature. Our objective is to address this gap by studying the effects of stochastic uncertainties on phase-coupled oscillators.

In phase-coupled models, oscillators are interconnected via nonlinear periodic functions. The most studied model of phase-coupled oscillators is the well-known Kuramoto model~\cite{strogatz2000kuramoto}, in which the coupling law is reduced to its leading term in an odd Fourier expansion, \ie, sine function. The continuous-time deterministic Kuramoto model has been widely employed to study synchronization in various applications~\cite{jadbabaie2004stability,delabays2019kuramoto,dorfler2014synchronization,franci2011existence,jafarian2018sync,menara2019stability}. Besides, the importance of more general classes of coupling functions on synchronization has also been shown in the literature (see~\cite{okuda1993variety,watanabe1997stability,brown2003globally,nishikawa2004oscillatory}, and references therein).  

We study the effects of stochastic multiplicative and additive uncertainties in a network of phase oscillators coupled via $2\pi$-periodic odd functions. Odd coupling functions generalize the Kuramoto model, by allowing the  addition of some higher Fourier harmonics (e.g. phase model of LC neurons~\cite{brown2003globally}), and present a gradient dynamical system, which has been proven useful in achieving convergence in oscillatory networks~\cite{acebron2005kuramoto,izhikevich1999weakly,mallada2013synchronization}. For instance, phase oscillators coupled via odd functions have been used in studying the associate memory feature of integrate-and-fire spiking neuronal networks~\cite{izhikevich1999weakly}, as well as recall of oscillatory associative memory networks~\cite{nishikawa2004oscillatory}. The combination of generality, applicability and mathematical amenability inspire us to primarily focus on this class of coupling functions. Furthermore, we discuss an extension of our results by allowing a relaxation of the odd property of the coupling function.

This article studies stochastic set stability for phase oscillators described above, using a discrete-time approximation, and perform the analyses within the framework of Markov chain stability \cite{kushner1990numerical,meyn2012markov}. 
%----------------------------------------------------------------
\subsection{Related works}
The effects of uncertainties on the synchronization of oscillatory networks have mainly been investigated in continuous-time deterministic perturbed models, and stochastic models with only additive noise. In the stochastic setting, the effects of additive noise, such as uncertain exogenous frequencies, have been studied for some classes of continuous-time nonlinear oscillators using a phase reduction model \cite{teramae2004robustness}, and a Fokker-Planck model \cite{acebron2005kuramoto,bag2007influence}. Moreover, synchronization of continuous-time Kuramoto oscillators with additive uncertainties, modeled as a Wiener process, has been analyzed in a stochastic game setting \cite{yin2011synchronization}. A mean-filed approach to frequency synchronization of continuous-time coupled Kuramoto oscillators subject to random exogenous frequencies has been studied in \cite{ichinomiya2004frequency}. Oscillators in Erd\"{o}s-R\'{e}nyi \cite{random,hatano2005agreement} random networks, in which each two oscillators are decoupled with a non-zero probability, have also been studied \cite{manaffam2013synchronization,preciado2005synchronization,porfiri2008synchronization,
zhou2006universality}. All aforementioned studies have focused on the asymptotic full-state synchronization of continuous-time oscillators in the Euclidean space using a linear approximation of the network dynamics in the vicinity of the synchronous state. 

The effects of coupling noise have been studied on the full-state synchronization of a network of continuous-time harmonic oscillators \cite{wang2018stochastic}, and interconnected oscillators with a common or zero intrinsic noise together with state-dependent coupling noise~\cite{aminzare2022stochastic,russo2018synchronization}. These works have analyzed the network behavior using stochastic differential equations and discussed the effects of noise on the synchronization.
 
Deterministic perturbed Kuramoto model has been studied considering time-varying exogenous frequencies and coupling coefficients \cite{franci2010phase,lu2018stability,zhu2020synchronization}. The local input-output stability of the exact synchronization (phase-locking) solution has been proven imposing an upper-bound on the time-varying exogenous frequencies \cite{franci2010phase}. Moreover, phase-cohesiveness of the in-phase solution in a network with time-varying couplings and exogenous frequencies \cite{lu2018stability} has been shown imposing conditions on the time-evolution of the coupling coefficients. Compared with the literature, our contributions, detailed below, include: studying stochastic phase-coupled oscillators; generalizing the coupling function; considering stochastic multiplicative and additive uncertainties (see Remark \ref{re}); and studying the stability of phase-coupled oscillators in random networks. It is worth noting that deterministic unperturbed discrete-time Kuramoto oscillators, with applications in communication and robotics networks, have also been studied \cite{favaretto2017cluster,klein2008integration}. Our construction of Lyapunov function for the general connected and undirected networks, extends the analysis of the aforementioned works, if we replace the stochastic variables with the deterministic ones.
%-------------------------------------------------------------
\subsection{Main contributions}
Our objective is to study the effects of multiplicative, \ie, system's states or a function of them multiplied by stochastic random variables \cite{willems1976feedback}, and additive uncertainties on the stability of relative phases of discrete-time phase-coupled oscillators. Our main contributions are highlighted below.\\[1mm]
First, we introduce the new notions of stochastic phase-cohesiveness and ultimate stochastic phase-cohesiveness. The notion of phase-cohesiveness for deterministic models is defined based on the concept of invariant sets \cite{dorfler2014synchronization}. For stochastic models, we use the concept of Harris recurrent Markov chains. Basically, a network of stochastic phase-coupled oscillators is phase-cohesive with respect to a desired set if the probability that the relative phases return to this set, after leaving it, is one (see Definition~\ref{def1}). Ultimate stochastic phase-cohesiveness indicates a bounded return time (see Definition~\ref{def2}).\\[1mm]
Second, we study the phase-cohesiveness of discrete-time phase oscillators over an undirected graph subject to both multiplicative (\ie, coupling weights) and additive (\ie, exogenous frequencies) stochastic uncertainties. The presence of multiplicative uncertainties indicates that the underlying topology is connected in a probabilistic sense. We consider $2\pi$-periodic, bounded, and odd coupling laws, and independent and identically distributed (i.i.d) random uncertainties obeying normal distributions. In Theorem \ref{th1}, we assume either strictly positive or negative mean values for the uncertain couplings. By obtaining sufficient conditions in the form of lower bounds on the common coupling coefficient, we prove that depending on the sign of the mean values of multiplicative uncertainties, the uncertain network achieves stochastic phase-cohesiveness with respect to either in-phase or anti-phase sets. These sets are defined in the vicinity of the roots of the coupling function, \ie, zero (in-phase) and $\pi$ (anti-phase). We also characterize the conditions for achieving ultimate phase-cohesiveness (see Corollary~\ref{th2}). Respecting the discrete-time setting, our results are derived assuming a sufficiently small sampling-time for which we also characterize an upper bound.

In Proposition~\ref{co:clus}, we then allow the co-existence of edges whose corresponding multiplicative uncertainties possess either negative or positive mean values. Assuming identical exogenous frequencies, we study this special case over an underlying line topology. Conditions under which the network exhibits clustering behavior are derived by showing its stochastic phase-cohesiveness with respect to the union of the in-phase and anti-phase sets.\\[1mm]
Third, we discuss a relaxation of the odd property of the coupling function. We allow the coupling function to be non-odd only on a subset of its domain. This relaxation describes a more general coupling law which is applicable, for instance, in studying series arrays of Josephson junctions~\cite{watanabe1997stability}. For oscillators over an undirected tree network with multiplicative (positive mean-values) and additive uncertainties, conditions under which the network is stochastic phase-cohesive with respect to the in-phase set are characterized (see Proposition~\ref{co:sh}).\\[1mm]
Fourth, we study phase-coupled oscillators, with $2\pi$-periodic odd couplings, such that each two coupled oscillators are connected using a common and constant coupling term and with a non-zero probability of connection, which represents an Erd\"{o}s-R\'{e}nyi random network. We obtain conditions for achieving both stochastic phase-cohesive and phase-locked solutions for this network. For the case of constant and non-identical frequencies, we show that the effect of the randomness is to stabilize the in-phase set as the only absorbing set of the network (see Theorem~\ref{th3}). Furthermore, we prove that any positive probability of connection will lead to a phase-locking solution when all oscillators have identical exogenous frequencies (see Corollary~\ref{th4}). 

In Remark \ref{re}, stochastic analysis is compared with deterministic perturbation analysis. To the best of our knowledge, the stochastic stability of phase-coupled oscillators has not been studied in the literature. A preliminary result for Kuramoto oscillators over a tree network was presented in our conference paper \cite{jafarian2019stochastic}.
\subsection{Outline}
The rest of this article is organized as follows. Section \ref{sec:pre} provides the model and the required preliminaries. Section \ref{sec:pf} presents the problem formulations, and the notions of stochastic phase-cohesiveness, and ultimate stochastic phase-cohesiveness. The analyses of uncertain and random networks are presented in Sections \ref{sec:noise} and \ref{sec:rand}, respectively. Section \ref{sec:sim} presents simulation results and Section \ref{sec:cl} summarizes the concluding remarks. All proofs are provided in the appendices.\\[1mm]
%------------------------------------------------------------------------------------------------------------------------------------------------
\noindent{\bf{Notation}:}
The notation $x_{i,j}$ is equivalently used for $x_i-x_j$. A random variable $x$ selected from an arbitrary distribution $\mathcal{X}$ with mean $\mu$ and variance $\sigma^2$ is denoted by $x \sim \mathcal{X} (\mu, \sigma^2)$. The expected value and conditional expected value operators are denoted by $\boldsymbol E[\cdot]$ and $\boldsymbol E[\cdot|\cdot]$, respectively. The symbol $\S^1$ denotes the unit circle. An angle is a point in $\S^1$ and an arc is a connected subset of $\S^1$. Given a matrix $M$ of real numbers, we denote by ${\cal R} (M)$ and ${\cal N} (M)$ the range and the null space, respectively. Symbol $\mathbf{1}_n$ is a $n$-dimensional vector of all ones. The empty set is denoted by $\emptyset$. 
%================================================================================================================================================
\section{Model and Preliminaries}\label{sec:pre} 
This article considers a network of discrete-time phase-coupled oscillators governed by dynamics in \eqref{eq:p1}, which is a discrete-time approximation of its continuous-time counterpart
\be\begin{aligned}\label{eq:p1}
{\theta_i}(\mathrm{k}+1)=&\Big\{{\theta_i}(\mathrm{k}) + \tau{\tilde\omega}_i(\mathrm{k})- \\ & \Big(\kappa\tau \sum_{j \in {\cal N}_i} {\tilde\alpha}_{i,j}(\mathrm{k}) \Psi(\theta_{i,j}(\mathrm{k}))\Big)\Big\}\hspace{-3mm}\pmod{2\pi},
\end{aligned}\ee
where $\mathrm{k}\in \mathbb{Z}^+$, $\theta_i(\mathrm{k}) \in \S^1$, $\theta_{i,j}(\mathrm{k})$, $\tau >0$, and $\kappa >0$ represent the time step, the phase of oscillator $i$, the relative phase of oscillators $i$ and $j$, the sampling time, and the common coupling term, respectively. The set of neighbors of oscillator~$i$ is denoted by ${\cal N}_i$. The variables ${\tilde\alpha}_{i,j}(\mathrm{k})$ and ${\tilde\omega}_i(\mathrm{k})$ are stochastic variables representing the uncertain multiplicative coupling weight, and the uncertain exogenous frequency, respectively. Function $\Psi(\cdot)$ is the coupling function. For the purpose of brevity, the term $\pmod{2\pi}$ is omitted from all representations of the discrete-time dynamics in the rest of the article. 

\begin{assumption}\label{as1}
Function $\Psi(\cdot)$~is $2 \pi$-periodic, continuously differentiable, and odd\footnote{The relaxation of this assumption is discussed in Section \ref{sec:noise}.}. In the interval $[0,\pi]$, it holds that $\Psi(0)=\Psi(\pi)=0$, and $|\Psi(\xi)| \neq 0, \forall \xi \in (0,\pi)$. 
\end{assumption}
Since $\Psi(\cdot)$ is continuously differentiable, hence bounded, \ie, $|\Psi(\xi)| \leq \Psi_{\max}$, the following property immediately follows. This property indicates that there exits an arc $\Upsilon \subset [0,\pi]$ such that the value of $\Psi(\cdot)$ for every angle on this arc is greater than angles which belong to arcs in the vicinity of $0$ and $\pi$.
\begin{property}\label{prop1}
For any function $\Psi(\cdot)$ satisfying Assumption \ref{as1}, there exists an arc $\Upsilon=(\gamma,\gamma_{\max}) \subset (0,\pi)$, such that $\forall \xi, z \in [0,\pi]$, if $\xi \not\in \Upsilon, z \in \Upsilon$ then $|\Psi(\xi)| \leq |\Psi(z)|$. 
\end{property}

\subsection{Preliminaries}
The term $|\theta_i(\mathrm{k})-\theta_j(\mathrm{k})|$ denotes the {\em geodesic distance} between phases $\theta_i, \theta_j \in \S^1$. The geodesic distance is the minimum value between the counter-clockwise and the clockwise arc lengths connecting $\theta_i$ and $\theta_j$. The size of the relative phase $\theta_{i,j}= \theta_i-\theta_j \in (-\pi,\pi]$ equals $|\theta_i-\theta_j|$ and its sign is positive if the counter-clockwise path length from $\theta_i$ to $\theta_j$ is smaller than that of the clockwise path. The relative phase of the two oscillators decreases if 
$|{\theta_{i,j}}(\mathrm{k}+1)| < |{\theta_{i,j}}(\mathrm{k})|.$ 
%--------------------------------------------------------------------------
\subsection*{Graph theory} We here revisit some preliminaries on graph theory mainly borrowed from \cite{mesbahi2010graph}. Consider an undirected graph $G(\mathcal V,\mathcal E)$, where $\mathcal V$ is the set of $n$ nodes and $\mathcal E \subset \mathcal V \times \mathcal V$ is the set of $m$ edges. The graph's incidence matrix is denoted by $B_{n \times m }$. The two matrices $L(G) \triangleq B B^T$ and $L_e(G)\triangleq B^T B$ are called the graph Laplacian and the edge Laplacian, respectively. If the underlying graph is connected, then the eigenvalues of $L$ can be ordered as $0=\lambda_1(L) < \lambda_2(L)\leq ... \leq \lambda_n(L)$, where $\lambda_2(L)$ is called the {\it algebraic connectivity} of the graph. Moreover, all non-zero eigenvalues of $L_e$ are equal to the non-zero eigenvalues of $L$. 

A spanning tree of $G$ is a subgraph $G_\tau({\mathcal V},{\mathcal E}_{\tau})$ which is a tree (cycle free) graph. Under an appropriate permutation of the edge indices, the incidence matrix of a connected graph $G$ can be partitioned as $B=[B(G_\tau) \quad B(G_c)]$, where $G_\tau$ represents a given spanning tree of $G$, and $G_c$ represents the remaining edges. There exists a matrix $R$ such that $L_e(G)= R^\top L_e(G_\tau) R$, where $R=[I \quad T]$ with $T=L^{-1}_e(G) B^{\top}(G_\tau) B(G_c)$. For tree graphs, the edge Laplacian is positive definite.

\subsection*{Markov chains} 
A general measurable space is a pair $(X,\mathcal{B}(X))$ with $X$ a set of points and $\mathcal{B}(X)$ a $\sigma$-algebra of subsets of $X$ satisfying the following properties:
\begin{itemize}
\item[(a)] $\emptyset \in \mathcal{B}(X)$,
\item[(b)] If $D \in \mathcal{B}(X)$, then $D^c \in \mathcal{B}(X)$, where $D^c=X\setminus D$,
\item[(c)] If $D_1\in \mathcal{B}(X)$ and $D_2\in \mathcal{B}(X)$, then $D_1\cup D_2 \in \mathcal{B}(X)$. 
\end{itemize} 

A Markov chain is a stochastic process $\Phi=\{\Phi_0, \Phi_1, \ldots\}$ such that each $\Phi_i$ is randomly taking values on the measurable state-space $X$ which is endowed by the $\sigma$-algebra $\mathcal{B}(X)$. The chain $\Phi$ is defined by the triple $(\Omega, \mathcal{F}, \textsf{P})$, such that:
\begin{itemize}
\item $\Omega$ is the whole state-space, \ie, the product of all pairs, each of which corresponds to $\Phi_i$ and is a subset of $(X,\mathcal{B}(X))$,
\item $\mathcal{F}$ is a $\sigma$-algebra associated to the measurable space $\Omega$, 
\item $\mathsf{P}\!:\!\mathcal{F}\!\rightarrow \![0,1]$ is a probability measure defined on $(\Omega, \mathcal{F})$ that assigns a probability to each outcome of~$\mathcal{F}$.
\end{itemize} 
\vspace{2mm}
\noindent The following definitions are mainly borrowed from \cite{meyn2012markov}. 

\begin{definition}\cite[Ch.3]{meyn2012markov}\label{def:cont}
Let $(X,\mathcal{B}(X))$ be a measurable space. The state-space $X$ is called \textit{countable} if $X$ is discrete, with a countable number of elements, and with $\mathcal{B}(X)$ the $\sigma$-algebra of all subsets of $X$. The state-space $X$ is called \textit{general} if it is assigned a countably \textit{generated} $\sigma$-algebra $\mathcal{B}(X)$.\footnote{The smallest $\sigma$-algebra on which $\mathcal{B}$ is measurable, i.e., the intersection of all $\sigma$-algebras on which $\mathcal{B}$ is measurable, is called the generated $\sigma$-algebra by $\mathcal{B}$.} 
\end{definition}

\begin{definition}\label{def:markov}
For a stochastic process $\Phi\!=\!\{\Phi_0, \Phi_1, \ldots\}$ defined on $(\Omega, \mathcal{F}, \textsf{P})$, let $P^n(\omega,\mathcal{B})$ denotes the transition probability that the process enters the set $\mathcal{B}$ after $n$ transitions, i.e., $\Phi_{n+m}\!\in \!\mathcal{B}$, given $\Phi_m \!=\! \omega$. Then $\Phi$ is a \textit{time-homogeneous} Markov chain if transition probabilities $\{P^n(\omega,\mathcal{B}), \omega\in \Omega, \mathcal{B}\subset \Omega\}$ exist such that for any $n,m\in \mathbb{Z}^+$, the following holds:
\begin{equation*}
\mathsf{P}(\Phi_{n+m}\in \mathcal{B}|\Phi_j,j\leq m, \Phi_m=\omega)=P^n(\omega,\mathcal{B}).
\end{equation*}
The independence of the transition probability $P^n(\omega,\mathcal{B})$ from $j\leq m$ entails the Markov property, and its independence from $m$ confirms the time-homogeneity.
\end{definition}

\begin{definition}\label{def:irreducibility}
Let $\Phi\!=\!\{\Phi_0,\Phi_1,\ldots\}$ be a Markov chain defined on $(\Omega, \mathcal{F}, \textsf{P})$. Then:
\begin{enumerate}
\item for any $D\in \mathcal{F}$, the measurable function $\tau_D: \Omega \rightarrow \mathbb{Z}^+\cup \{\infty\}$ is the \textit{first return time} to the set $D$, i.e.,
\begin{equation}\label{eq:ret_time}
\tau_D \triangleq \min \{n\geq 1\;|\;\Phi_n\in D\},
\end{equation}
\item for any measure $\varphi$ on the $\sigma$-algebra $\mathcal{F}$, the Markov chain $\Phi$ is said to be \textit{$\varphi$-irreducible} if $\forall \omega \in \Omega$ and $D\in \mathcal{F}$, $\varphi(D)>0$ implies
$\mathsf{P}(\tau_D<\infty)>0$.
\end{enumerate}
\end{definition}
According to Definition \ref{def:irreducibility}, the entire state-space of a Markov chain is reachable, independent of the initial state, via finite number of transitions if the Markov chain is $\varphi$-irreducible. Moreover, in that case, a unique maximal irreducibility measure $\psi > \varphi$ exists on $\mathcal{F}$ such that $\Phi$ is $\varphi^\prime$-irreducible for any other measure $\varphi^\prime$ if and only if $\psi > \varphi^\prime$. We then say that the Markov chain is $\psi$-irreducible.
\begin{definition}\cite[Ch.5]{meyn1994}
Let $a=\{a(n)\}$ be a probability measure on $\mathbb{Z}^+$ and $\Phi_a$ be the sampled chain of the Markov chain $\Phi$ at time-points drawn successively according to the distribution~$a$. Denote the probability transition kernel of $\Phi_a$ by $K_a(\omega,D)=\sum_{n=0}^{\infty} K^n(\omega,D)\;a(n)$, where $K(\omega,D)$ and $K^n(\omega,D)$ are the probability transition kernel of $\Phi$, and the $n$-step probability transition kernel of $\Phi$, respectively. A set $C\in \mathcal{B}$ is a $\nu_a$-petite set, where $\nu_a$ is a non-trivial measure on $\mathcal{B}$, if for all $\omega\in C$ and $D\in \mathcal{B}$, $\Phi_a$ satisfies %the following condition:
\begin{equation*}
K_a(\omega,D)\geq \nu_a(D).
\end{equation*}
\label{def:petite_set}
\end{definition}
\begin{definition}\cite[Ch. 6]{meyn1994}
Let $K$ be the transition probability kernel of a chain $\Phi$ defined on a locally compact and separable space $X$ acting on a bounded function $h(x): X\rightarrow \mathbb{R}, x\in X$ via the mapping $K h(x)=\int K(x,dy)h(y)$. Denoting the class of bounded continuous functions from $X$ to $\mathbb{R}$ by $\mathcal{C}(X)$, the chain $\Phi$ has the Feller property if $K$ maps $\mathcal{C}(X)$ to $\mathcal{C}(X)$.
\label{def:feller1}
\end{definition}
%%=============================================================================================================
\begin{table}
\begin{center}
   \caption{Most used variables and notations}
\begin{tabular}{c|l}
\hline
  $\theta_i(\mathrm{k})$  & \text{Phase of oscillator} $i$ {at time-step} $\mathrm{k}$\\[0.5mm]
  $\theta_{i,j}(\mathrm{k})$ & Relative phase, \ie, $\theta_i(\mathrm{k})-\theta_j(\mathrm{k})$\\[0.5mm]
   $\boldsymbol\theta(\mathrm{k})$ & Augmented phase vector \\[0.5mm]
   $\Theta(\mathrm{k})$ & Augmented relative phase vector \\[0.5mm]
   $\S^1$ & Unit circle\\[0.5mm]
    $\kappa$ & Common coupling strength\\[0.5mm]
    $\tau$ & Common sampling time\\[0.5mm]
    $\Psi$ & Periodic odd coupling function\\[0.5mm]
    $\tilde{\alpha}_{i,j}(\mathrm{k})$ & Multiplicative uncertainty of edge $(i,j)$\\[0.5mm] 
    $\tilde{\omega}_{i}(\mathrm{k})$ & Uncertain exogenous freq. of oscillator $i$\\[0.5mm] 
    $\mathcal{E}$ & Edge set of the underlying network graph\\[0.5mm]
    $\mathcal{N}_i$ & Set of neighbors of oscillator $i$\\[0.5mm]
    $\Pi$ & State space of augmented relative phases\\[0.5mm]
    $\gamma$ & Phase angle on the arc $[0,\pi]$\\[0.5mm]
    $S^{\Psi}_{G}(\gamma)$ & In-phase set \\[0.5mm]
    $\tau_{S^{\Psi}_{G}(\gamma)}$ & First return time to in-phase set \\[1mm]
    $S^{\Psi}_{G}(\gamma_{\max})$ & Anti-phase set \\[1mm]
    \hline
\end{tabular}
\end{center}
\end{table}
%================================================================================================================================================================
\section{Problem formulation}\label{sec:pf}
This section presents definition of the stochastic phase-cohesiveness, and presents the problem statement. 
%------------------------------------------------------------------------------------------------------------
\subsection{Stochastic phase-cohesiveness}
This section introduces two new definitions for stochastic stability of phase-coupled oscillators based on two notions of stability of stochastic processes. Our first definition, stochastic phase-cohesiveness, corresponds to the concept of Harris recurrent Markov chains \cite[Ch. 9]{meyn2012markov}. A $\psi$-irreducible chain with state space $X$ is Harris recurrent if it visits every set $A \in {\mathcal B}^{+}(X)$ almost infinitely, where ${\mathcal B}^{+}(X)\triangleq \{A\in \mathcal{B}(X): \psi(A)>0\}$ and $\psi$ is the maximal irreducibility measure.\footnote{The set ${\mathcal B}^{+}(X)$ contains all sets of positive $\psi$-measure subsets of ${\mathcal B}(X)$ and is uniquely defined for $\psi$-irreducible chains \cite[Ch.4]{meyn2012markov}.} Equivalently, a chain is Harris recurrent if the probability of its first return time $\tau_A$ to a desired set $A$ is one \cite{meyn2012markov}.

Let $\Theta$ denote the augmented relative phase vector of the interconnected oscillators with the dynamics in \eqref{eq:p1}: 
\be\label{eq:T} \Theta(\mathrm{k})={\mathcal H}^{G} (\boldsymbol\theta(\mathrm{k})),\ee where $\boldsymbol\theta(\mathrm{k})\triangleq [\theta_1(\mathrm{k}),\ldots,\theta_n(\mathrm{k})]^\top$ is the augmented phase state, and ${\mathcal H}^{G}$ is a topology-dependent operator which computes the relative phases. The state space on which $\Theta(\mathrm{k})$ evolves is defined by 
\begin{align}\label{eq:si}
\Pi=[-\pi,\pi]^{\|\mathcal{E}\|},
\end{align}
where $\|\mathcal{E}\|$ denotes the number of edges of the underlying network topology.

\begin{definition}[Stochastic phase-cohesiveness]\label{def1}
The~relative~phase process $\Theta(\mathrm{k})$ in \eqref{eq:T}, is stochastic phase-cohesive if $\textsf{P}_{\Theta(\mathrm{k})}(\tau_{S^{\Psi}_{G}(\gamma)} <\infty)=1$, where $S^{\Psi}_{G}(\gamma) \subset \Pi$ is a desired compact set, and $\tau_{S^{\Psi}_{G}(\gamma)}$ is the first return time of the stochastic process $\Theta(\mathrm{k})$ to the set $S^{\Psi}_{G}(\gamma)$.
\end{definition}

The above definition requires that the Markov chain returns to a desired set, $S^{\Psi}_{G}(\gamma)$, almost infinitely. In what follows, we present a stronger notion based on the concept of positive Harris recurrent chains \cite{meyn2012markov}. Conceptually, the below definition is the stochastic counterpart of deterministic ultimate boundedness. It implies that not only the chain revisits $S^{\Psi}_{G}(\gamma)$ almost infinitely but also the return time is bounded and the chain probability distribution converges to an invariant probability measure in a stationary regime. 

\begin{definition}[Ultimate stochastic phase-cohesiveness]\label{def2}
The~relative phase process $\Theta(\mathrm{k})$ in \eqref{eq:T} is ultimate stochastic phase-cohesive if there exists a constant $M < \infty$ such that $\sup \boldsymbol E_x[\tau_{S^{\Psi}_{G}(\gamma)}] \leq M$ for all $x \in S^{\Psi}_{G}(\gamma)$, where $\boldsymbol E_x$ denotes the expectation of events conditional on the chain beginning with $\Phi_0=x$.
\end{definition}
%----------------------------------------------------------------------------------------------------------------------
\subsection{In-phase and anti-phase sets}\label{sec:set}
Here, we introduce two sets, namely in-phase and anti-phase, which are useful in our problem statement. For a function $\Psi(\cdot)$, satisfying Assumption \ref{as1} and Property \ref{prop1}, choose $\gamma$ and $\gamma_{\max}$ such that $0 < \gamma < \gamma_{\max} < \pi$, $0 <|\Psi(\gamma)|=|\Psi(\gamma_{\max})| < \Psi_{\max}$ hold. Define the arcs $\underline{\Upsilon}$, ${\Upsilon}$, and $\overline{\Upsilon}$ as follows:
\be\label{eq:a1} 
\hspace{-10mm}\underline{\Upsilon}=[0,\gamma]: \forall \gamma_i \in \underline{\Upsilon}: 0 \leq |\Psi({\gamma}_i)| \leq |\Psi(\gamma)|,\hspace{9mm}
\ee
\be\label{eq:a2} 
{\Upsilon}=(\gamma,\gamma_{\max}): \forall \gamma_i \in {\Upsilon}: |\Psi(\gamma)| < |\Psi({\gamma}_i)| < \Psi_{\max},\hspace{7mm}
\ee
\be\label{eq:a3} 
\hspace{-15mm}\overline{\Upsilon}=[\gamma_{\max},\pi]: \forall \gamma_i \in \overline{\Upsilon}: 0 \leq |\Psi({\gamma}_i)| \leq |\Psi(\gamma)|.
\ee

\begin{figure}[t!]
    \centering
    \includegraphics[width=0.48\textwidth]{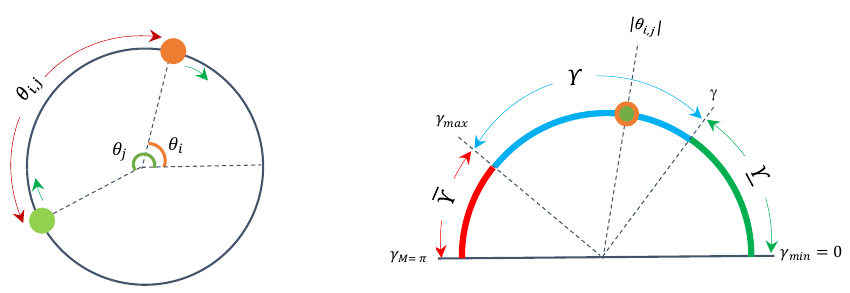}
    \caption{The relative phase of two oscillators $i$ and $j$ and its map on the arcs $\underline{\Upsilon}$, ${\Upsilon}$, and $\overline{\Upsilon}$.}
    \label{fig:f1}
\end{figure}
Figure \ref{fig:f1} depicts the introduced arcs together with the relative phase of two oscillators $i$ and $j$.\\[1mm]
\noindent We introduce the in-phase set, $S^{\Psi}_{G}(\gamma) \subset \Pi$, and anti-phase, $S^{\Psi}_{G}(\gamma_{\max}) \subset \Pi$, as
\be\label{eq:s1}
\hspace{3mm} S^{\Psi}_{G}(\gamma)=\{\theta_i \in \S^1, \theta_j \in \S^1:|\theta_{i,j}(\mathrm{k})| \in \underline{\Upsilon}, \forall (i,j) \in \mathcal{E}\}, 
\ee
\be\label{eq:s4}
S^{\Psi}_{G}(\gamma_{\max})=\{\theta_i \in \S^1\!\!,  \theta_j \in \S^1\!\!:|\theta_{i,j}(\mathrm{k})| \in \overline{\Upsilon}, \forall (i,j) \in \mathcal{E}\}.\hspace{2mm} 
\ee
%-----------------------------------------------------------------------------------------------------------
\subsection{Problem statement}\label{sec:ps}
Here, we formulate the two problems to be studied in Sections \ref{sec:noise} and \ref{sec:rand}, respectively.\vspace{-2mm}
\begin{problem}\label{p1} Consider a network of $n$ phase-coupled oscillators where the dynamics of each oscillator obeys \eqref{eq:p1}. Let 
\be\label{eq:pp1} {\tilde\alpha}_{i,j}(\mathrm{k})= {\alpha}_{i,j} + n_{i,j}(\mathrm{k}),\quad {\tilde\omega}_i(\mathrm{k})= \omega_i + {\varpi}_i(\mathrm{k}),\ee
where ${\alpha}_{i,j}, \omega_i \in \R^{+}$ are constants, $n_{i,j}(\mathrm{k})$ and ${\varpi}_i(\mathrm{k})$ are i.i.d. stochastic variables at each time step $\mathrm{k}$ selected from
a continuous distribution with finite mean and variance. Also, assume that the initial relative phase $\theta_{i}(0)-\theta_{j}(0)$ is an arbitrary random variable, independent from the realizations of $n_{i,j}(\mathrm{k})$ and ${\varpi}_i(\mathrm{k})$, $\forall \mathrm{k}$ and $\forall i,j$, with a finite moment probability distribution with a continuous density function.

Denote the relative phase stochastic process corresponding to this network by $\Theta(\mathrm{k})$ as in \eqref{eq:T}. Our objective is to study the stability of $\Theta(\mathrm{k})$, especially, to characterize coupling,~$\kappa$, condition under which the process is stochastic phase-cohesive with respect to the in-phase and anti-phase sets defined in \eqref{eq:s1} and \eqref{eq:s4}. 
\end{problem} \vspace{-5mm}
%-------------------------------------------------------------------------------------------------------------------------
\begin{problem}\label{p2} Now, consider a special case of the network in Problem \ref{p1} by assuming that the multiplicative uncertainties obey the Bernoulli distribution. That is, two oscillators $i$ and $j$ are coupled with probability $p$ and decoupled with probability $1-p$. When two oscillators are coupled, the weight of their corresponding edge equals $\kappa >0$. We assume constant exogenous frequencies. To make a distinction with Problem \ref{p1}, denote the multiplicative uncertainty by $\beta_{i,j}$. The dynamics of oscillator $i$ can then be expressed as
\begin{equation}\label{eq:p2}
{\theta_{i}}(\mathrm{k}+1)={\theta_{i}}(\mathrm{k})- \Big(\kappa \tau \sum_{j \in {\cal N}_i} \! {\beta}_{i,j}(\mathrm{k}) \Psi(\theta_{i,j}(\mathrm{k}))\Big)+ \tau {\omega}_{i},
\end{equation}
where $\omega_i \in \R^{+}$ represents the constant exogenous frequency of oscillator $i$.\\[1mm]
Our objective is to study the stability of this random network and obtain coupling conditions under which the relative phase stochastic process, $\Theta(\mathrm{k})$, is stochastic phase-cohesive with respect to the in-phase set. Moreover, we study conditions under which the phase-locked solution, \ie, phase-cohesiveness with respect to $S^{\Psi}_{G}(0)$, is achieved.
\end{problem}
%===============================================================================================================================================
\section{Oscillators in an uncertain network}\label{sec:noise}
This section studies Problem \ref{p1} to characterize conditions under which stochastic and ultimate stochastic phase-cohesiveness are achieved for the network. 

Considering the dynamics of individual oscillators in \eqref{eq:p1} and stochastic variables in \eqref{eq:pp1}, the relative phase dynamics of two interconnected oscillators $i$ and $j$ follows
\begin{align}\label{eq:p11}
{\theta_{i,j}}(\mathrm{k}+1)&= {\theta_{i,j}}(\mathrm{k})+ \tau {\tilde\omega}_{i,j}(\mathrm{k})
\\\nonumber &- \Big(\kappa \tau \sum_{\ell \in \{i,j\}} \sum_{e \in {\cal N}_\ell} {\tilde\alpha}_{\ell,e}(\mathrm{k}) \Psi(\theta_{\ell,e}(\mathrm{k}))\Big). 
\end{align}

\begin{assumption}\label{ass2}
The underlying deterministic network topology $G(\mathcal V,\mathcal E)$ is connected. The uncertain interconnections of oscillators are undirected, \ie, ${\tilde\alpha}_{i,j}(\mathrm{k})={\tilde\alpha}_{j,i}(\mathrm{k})$. Each multiplicative random variable obeys the normal distribution, \ie, ${\tilde\alpha}_{i,j} \sim {\mathcal N}(\mu_{i,j}, \sigma^2)$. For each of the nominal exogenous frequencies, it holds that $\boldsymbol E [\tilde\omega_i]>0$. 
\end{assumption}

Notice that identical variances for multiplicative uncertainties are not needed in our analysis. This assumption is only made for the sake of clarity of representation (see Theorem~\ref{th1}). From Assumption \ref{ass2}, the relative phase vector for the whole network, $\Theta(\mathrm{k})$ in \eqref{eq:T}, is equal to $B^{\top} {\boldsymbol\theta}(\mathrm{k})$, where $B$ is the incidence matrix of the underlying deterministic graph $G(\mathcal V,\mathcal E)$. The augmented relative phase dynamics can then be written in the following compact form
\begin{equation}\label{eq:rel}
B^\top {\boldsymbol\theta}(\mathrm{k}+1)= B^\top \Big(\boldsymbol\theta(\mathrm{k})+ \tau {\tilde{\boldsymbol \omega}}(\mathrm{k})- \tau \kappa B {\underaccent{\sim}{\boldsymbol \alpha}}(\mathrm{k}) \Psi(B^\top \boldsymbol\theta(\mathrm{k}))\Big), 
\end{equation}
where 
\begin{equation*}{\tilde{\boldsymbol \omega}}_{n \times 1}(\mathrm{k})= \big({\tilde\omega}_{1}(\mathrm{k}),\ldots,{\tilde\omega}_{n}(\mathrm{k})\big)^{\top},\end{equation*}
\begin{equation*}{\underaccent{\sim}{\boldsymbol \alpha}}_{m \times m}(\mathrm{k})= 
\begin{pmatrix}
    {\tilde\alpha}_{l_1}(\mathrm{k}) & 0 & \dots & 0 \\
    0 & {\tilde\alpha}_{l_2}(\mathrm{k}) & \dots & 0 \\
    \vdots & \vdots & \ddots & \vdots \\
    0 & 0 & \dots & {\tilde\alpha}_{l_m}(\mathrm{k})
  \end{pmatrix},
\end{equation*}
where $l_p \in \mathcal E$, with $p \in \{1,\ldots,m\}$, denotes the $p$-the edge of the underlying graph $G(\mathcal V,\mathcal E)$. 
%----------------------------------------------------------------------------------------------------------------------------------------------
\subsection{Markov properties}
Before studying the network's behavior, we discuss essential properties of the stochastic process described in \eqref{eq:rel}.
\begin{lemma}\label{mc1}
The relative phase stochastic process~\eqref{eq:rel} is a time-homogeneous and $\psi$-irreducible Markov chain evolving in a general space.
\end{lemma}
The above properties guarantee that a non-zero probability exists for the Markov chain to make a transition from any initial state to any state in the whole state space.
%-----------------------------------------------------------------------------------------------------------------------------------------------------------
In the following result, we show that every compact set in the state space $\Pi$ is also a petite set (Definition \ref{def:petite_set}). This equivalence is indeed helpful in studying stochastic stability of the Markov chain \eqref{eq:rel} as will be shown in Section \ref{sec:noiseA}.
\begin{lemma}\label{pr2}
Every compact set in the state space of the relative phase stochastic process~\eqref{eq:rel} is a petite set.
\end{lemma}
%--------------------------------------------------------------------------------------------------------------------------------------------
\subsection{Stochastic phase-cohesiveness}\label{sec:noiseA}
In this section, we study the phase-cohesiveness of the coupled oscillators modeled in \eqref{eq:rel} employing the mathematical tools for the stability analysis of Markov chains. 
%-------------------------------------------------------------------------------------------------------------------------------------------------------------- 
We first derive sufficient conditions under which the relative phase Markov chain is Harris recurrent and hence stochastic phase-cohesive. Based on Theorem 9.1.8 of \cite{meyn2012markov}, a $\psi$-irreducible chain $\Phi$ defined on a state space $X$ is Harris recurrent if there exists a petite set $C \subset X$, and a function $V: X \rightarrow \R^{+}$ which is unbounded off petite sets (\ie, all sub-level sets of $V$ are petite), such that the following drift condition is satisfied
\begin{equation}\label{eq:th}
\Delta V=\boldsymbol E[V(\Phi_{\mathrm{k}+1})|\Phi_{\mathrm{k}}=x]-V(x)< 0, \quad \;\forall \Phi_{\mathrm{k}} \in X \setminus C. 
\end{equation}
%----------------------------------------------------------------------------------------------------------------------------------------
Our results are based on an application of the mentioned theorem for our network problem setting. For the purpose of illustration, we first provide an example by analyzing the behaviour of a network of two coupled oscillators. 
%------------------------------------------------------------------------------------------------------------------------------------------
\begin{example}\label{exam1}
Consider a network of two oscillators modeled as the discrete-time Markov chain in \eqref{eq:rel}:
\begin{equation*}
{\theta_{1,2}}(\mathrm{k}+1)= \underbrace{\theta_{1,2}(\mathrm{k}) - 2 \tau \kappa\  {\tilde\alpha}_{1,2}(\mathrm{k}) \Psi(\theta_{1,2}(\mathrm{k}))}_{a}+ \tau \underbrace{{\tilde\omega}_{1,2}(\mathrm{k})}_{b}, 
\end{equation*}
with $\theta_{1,2}= \theta_{1} - \theta_{2}$, ${\tilde\omega}_{1,2}= {\tilde\omega}_{1}-{\tilde\omega}_{2}$, and ${\tilde\alpha}_{1,2} \sim N(\mu, \sigma^2)$. Define $\underline{\Upsilon}=[0,\gamma], {\Upsilon}=(\gamma,\gamma_{\max}), \overline{\Upsilon}=[\gamma_{\max}, \pi]$. We will show that $S^{\Psi}_{G}(\gamma)=\{\theta_i \in \S^{1}, \theta_{j} \in \S^{1}: |\theta_{1,2}|\in \underline{\Upsilon}\}$ is an absorbing set for the chain, $\ie$, the chain is Harris recurrent with respect to this set. We will further show that the chain is transient on $S^{\Psi}_{G}(\gamma_{\max})=\{\theta_i \in \S^{1}, \theta_{j} \in \S^{1}: |\theta_{1,2}|\in \overline{\Upsilon}\}$, meaning that, the probability that the chain revisits $S^{\Psi}_{G}(\gamma_{\max})$ infinitely often is zero.

First assume that ${\boldsymbol E}[{\tilde\alpha}_{1,2}]=\mu>\sqrt{\frac{2 \sigma^2}{\pi}}$. We verify the conditions under which the chain returns to $S^{\Psi}_{G}(\gamma)$ with probability one. Take $V=|\theta_{1,2}|$ and assume that $|\theta_{1,2}(\mathrm{k})| > \gamma$. Calculating the drift of $V$ based on \eqref{eq:th}, we obtain 
\begin{equation*}
\Delta V(\mathrm{k})=\boldsymbol E[V(\mathrm{k}+1)]-V(\mathrm{k}) = \boldsymbol E[|\theta_{1,2}(\mathrm{k}+1)|] - |\theta_{1,2}(\mathrm{k})|.
\end{equation*}
We have $\boldsymbol E[|\theta_{1,2}(\mathrm{k}+1)|] \leq \boldsymbol E[|a|]+ \tau \boldsymbol E[|b|]$. Using the formula for the folded normal distribution~\cite{leone1961folded} (see Lemma~\ref{lem1}, inequality~\eqref{eq:bond}), we obtain
$$\boldsymbol E[|a|] \leq |\boldsymbol E[a]|+ 2 \kappa \tau \sqrt{\frac{2 \sigma^2}{\pi}} |\Psi(\theta_{1,2}(\mathrm{k}))|.$$
First, assume $\boldsymbol E[a]>0$. Also, without loss of generality, assume that $\theta_{1,2}(\mathrm{k}) >0$, thus $\Psi(\theta_{1,2}(\mathrm{k}))>0$. We obtain $$\Delta V(\mathrm{k})\leq - 2 \kappa \mu \Psi(\theta_{1,2}(\mathrm{k}))+ 2 \kappa \sqrt{\frac{2 \sigma^2}{\pi}} \Psi(\theta_{1,2}(\mathrm{k}))+ {\boldsymbol E}[|b|].$$
Thus, $\Delta V <0$ holds if
\begin{equation}\label{eq:ex1}
\kappa > {\frac{{\boldsymbol E}[|{\tilde\omega}_{1,2}(\mathrm{k})|]}{2 \Psi(\theta_{1,2}(\mathrm{k})) (\mu-\sqrt{\frac{2 \sigma^2}{\pi}})}}.
\end{equation}

To obtain a lower bound for $\kappa$, upper-bound of ${\boldsymbol E}[{\tilde\omega}_{1,2}]$ and lower-bound of $\Psi(\theta_{1,2})$ are needed in \eqref{eq:ex1}. Since we assumed that at time $\mathrm{k}$, $|\theta_{1,2}(\mathrm{k})| > \gamma$, then either $|\theta_{1,2}(\mathrm{k})| \in {\Upsilon}$ or $|\theta_{1,2}(\mathrm{k})| \in \overline{\Upsilon}$. With the former, we have $\min \Psi(\theta_{1,2})=|\Psi(\gamma)|$. However, for the latter case, $\min \Psi(\theta_{1,2})=0$ holds, which leads to the requirement of $\kappa \rightarrow \infty$ to guarantee a negative drift. Now, let us set $\min \Psi(\theta_{1,2})=|\Psi(\gamma)|$ in \eqref{eq:ex1}. As discussed above, the obtained $\kappa$ is not sufficiently large to guarantee the return of the chain from $\overline{\Upsilon}$, in this example equivalent to $S^{\Psi}_{G}(\gamma_{\max})$, to ${S}^{G}(\gamma)$. The question is whether $S^{\Psi}_{G}(\gamma_{\max})$ is an absorbing set. To verify, take $\bar{V}=\pi-|\theta_{1,2}|$ and assume that at time $\mathrm{k}$, $|\theta_{1,2}| < \gamma_{\max}$. Calculating the drift of $\bar{V}$, we obtain
$
\Delta {\bar V}(\mathrm{k})= |\theta_{1,2}(\mathrm{k})|-\boldsymbol E[|\theta_{1,2}(\mathrm{k}+1)|].
$
The above gives $\Delta {\bar V}(\mathrm{k})= -\Delta V(\mathrm{k})$. Thus, finding a bound for $\kappa >0$ to guarantee $\Delta {\bar V} <0$ is not possible. In particular, the lower bound for $\kappa$ in \eqref{eq:ex1}, with $\min \Psi(\theta_{1,2})=|\Psi(\gamma)|$, guarantees $\Delta V <0$, and thus $\Delta \bar{V} >0$. Based on Theorem 8.0.2 of \cite{meyn2012markov}, the chain is transient on $S^{\Psi}_{G}(\gamma_{\max})$. This implies that the probability that the chain revisits $S^{\Psi}_{G}(\gamma_{\max})$ infinitely often is zero. This example shows the positive effect of the additive uncertainties in the stabilization of the chain with respect to the in-phase set, $S^{\Psi}_{G}(\gamma)$. Notice that for ${\boldsymbol E}[{\tilde\alpha}_{1,2}]<0$, based on a similar argument, we can prove the stability of the chain with respect to the anti-phase set, $S^{\Psi}_{G}(\gamma_{\max})$. 
\end{example}
%---------------------------------------------------------------------------------------------------------------------------------------
We are now ready to state our main result.
\begin{theorem}\label{th1}
Consider the relative phase stochastic process \eqref{eq:rel} under Assumptions \ref{as1} and \ref{ass2}. The following statements hold:
\begin{enumerate}
\item if ${\boldsymbol E}[{\underaccent{\sim}{\boldsymbol \alpha}}]$ is positive definite, then the relative phase process is stochastic phase-cohesive with respect to the in-phase set, $S^{\Psi}_{G}(\gamma)$ defined in \eqref{eq:s1}, provided that $\mu_m > \sqrt{\frac{2 \sigma^2}{\pi}}$,
\begin{subequations}\label{eq:kappa-11}\be
\kappa > \frac{{\boldsymbol E}_{\max} [|\Delta{\tilde{\omega}}|]}{\big(\mu_m - \sqrt{\frac{2 \sigma^2}{\pi}}\big) |\Psi(\gamma)| \lambda_{\min}(L_{e}({G}_{\tau}))}, \ee
and $\tau$ satisfies
\be\begin{aligned} \tau < \frac{\gamma}{\kappa \big(\mu_M + \sqrt{\frac{2 \sigma^2}{\pi}}\big) \Psi_{\max} \lambda_{\max}(L_e) + {\boldsymbol E}_{\max} [|\Delta{\tilde{\omega}}|]},
\end{aligned}\ee\end{subequations}
\item if ${\boldsymbol E}[{\underaccent{\sim}{\boldsymbol \alpha}}]$ is negative definite and $B^{\top} \boldsymbol\theta \geq \gamma_{\max} {\mathbf 1}_{m}$ is feasible for $G(\mathcal V,\mathcal E)$, then the relative phase process is stochastic phase-cohesive with respect to the anti-phase set, $S^{\Psi}_{G}(\gamma_{\max})$ defined in \eqref{eq:s4}, provided that 
\begin{subequations}\label{eq:kappa-12}
\be \kappa > \frac{|{\boldsymbol E}_{\max} [\Delta{\tilde{\omega}}]|}{|\Psi(\gamma)|\  \mu_{m}\  \lambda_{\min}(L_{e}({G}_{\tau}))},\ee
and $\tau$ satisfies
\be\begin{aligned} \tau < \frac{(\pi-\gamma_{\max})}{\kappa \Psi_{\max}\ \mu_{M}\ \lambda_{\max}(L_e)\ + {\boldsymbol E}_{\max} [|\Delta{\tilde{\omega}}|]},
\end{aligned}\ee
\end{subequations}
\end{enumerate}
where $\mu_m=\lambda_{\min}(|{\boldsymbol E}[{\underaccent{\sim}{\boldsymbol \alpha}}]|)$, $\mu_M=\lambda_{\max}(|{\boldsymbol E}[{\underaccent{\sim}{\boldsymbol \alpha}}]|)$, ${\boldsymbol E}_{\max} [\cdot]\!=\! {\underaccent{i,j}{\max}}{\boldsymbol E}[\cdot]$, $\lambda_{\max}(L_e)$ is the largest eigenvalue of $L_e(G)$, $\lambda_{\min}(L_{e}({G}_{\tau}))$ is the minimum among the smallest eigenvalues of all spanning trees of $G$. \hfill \QED
\end{theorem}
%----------------------------------------------------------------------------------------------------------------------------------
{\bf{Sketch of the proof:}} The proof is based on an application of Theorem 9.1.8 of \cite{meyn1994} in a network setting. The key is the construction of a positive and radially unbounded function $V: \Pi \rightarrow \R^{+}$ for the discrete-time network, such that the one-step drift of $V$ is negative if sufficient coupling conditions are satisfied. The detailed proof is provided in Appendix~\ref{ath1}.\hfill \QED\\[1mm] 
%-----------------------------------------------------------------------------------------------------------------------
Theorem~\ref{th1} presents the coupling conditions under which the relative phase process is stochastic phase-cohesive. The main condition is the lower bound on $\kappa$. The bound on $\tau$ ensures that the sampling time is sufficiently small such that given a sufficiently large $\kappa$, the expectation of the maximum relative phase at each time-step is confined within the desired arc. We now continue by characterizing the coupling conditions under which ultimate phase-cohesiveness is achieved. We show that for the ultimate case, the coupling condition depends on $\tau$. Without loss of generality, the following result is presented for the case of ${\boldsymbol E}[{\underaccent{\sim}{\boldsymbol \alpha}}]>0$.

\begin{corollary}[Ultimate stochastic phase-cohesiveness]\label{th2}
Consider the discrete-time Markov chain in \eqref{eq:rel} representing the relative phase dynamics of $n$ interconnected oscillators under Assumptions \ref{as1} and \ref{ass2}. Assume that ${\boldsymbol E}[{\underaccent{\sim}{\boldsymbol \alpha}}] >0$ holds. Then, the relative phase process is ultimate stochastic phase-cohesive with respect to the in-phase set, $S^{\Psi}_{G}(\gamma)$ in \eqref{eq:s1}, if the following conditions hold:
\begin{subequations}\label{eq:kappa-2}
\be \kappa > \frac{\frac{1}{\tau |\Psi(\gamma)|}+ {\boldsymbol E}_{\max} [|\Delta{\tilde{\omega}}|]}{\big(\mu_m - \sqrt{\frac{2 \sigma^2}{\pi}}\big) |\Psi(\gamma)| \lambda_{\min}(L_{e}({G}_{\tau}))},\ee
\be\begin{aligned} \tau < \frac{\gamma- \frac{1}{m \Psi_{\max}}}{\kappa \big(\mu_M + \sqrt{\frac{2 \sigma^2}{\pi}}\big) \Psi_{\max} \lambda_{\max}(L_e) + {\boldsymbol E}_{\max} [|\Delta{\tilde{\omega}}|]},\hspace{3mm}
\end{aligned}\ee\end{subequations}
where $\mu_m=\lambda_{\min}(|{\boldsymbol E}[{\underaccent{\sim}{\boldsymbol \alpha}}]|)$, $\mu_M=\lambda_{\max}(|{\boldsymbol E}[{\underaccent{\sim}{\boldsymbol \alpha}}]|)$, ${\boldsymbol E}_{\max} [\cdot]\!=\! {\underaccent{i,j}{\max}}{\boldsymbol E}[\cdot]$, $\lambda_{\max}(L_e)$ is the largest eigenvalue of $L_e(G)$, $\lambda_{\min}(L_{e}({G}_{\tau}))$ is the minimum among the smallest eigenvalues of all spanning trees of $G$.  
\hfill \QED
\end{corollary}
%======================================================================================
\subsection{Mixed positive and negative multiplicative mean values}
Theorem \ref{th1} has proved that the relative phase Markov chain is stochastic phase-cohesive with respect to either the in-phase set, $S^{\Psi}_{G}(\gamma)$, or anti-phase set $S^{\Psi}_{G}(\gamma_{\max})$ depending on the sign of the mean values of the multiplicative uncertainties. The presence of mixed positive and negative mean values for the uncertain couplings could lead to a positive or negative drift condition. In this section, we study the stochastic stability of a network with an underlying connected and undirected line topology, \ie, a subclass of connected graphs without cycles in which each oscillator is connected to maximum two other oscillators. We obtain conditions under which all relative phases with positive multiplicative mean-values are recurrent to the in-phase set and those with negative mean-values are recurrent to the anti-phase set. We consider identical exogenous frequencies and assume zero mean value for the additive uncertainties, \ie, ${\boldsymbol E}[{\tilde\omega}_{i,j}(\mathrm{k})]=0, \forall (i,j)$. Recall the arcs $\underline{\Upsilon}$ and $\overline{\Upsilon}$, defined in \eqref{eq:a1} and \eqref{eq:a3}, respectively. Define,
\be\label{eq:s2}
U^{\Psi}_{G}(\gamma)=\{\theta_i \in \S^1\!\!, \theta_j \in \S^1\!\!:|\theta_{i,j}(\mathrm{k})| \in \underline{\Upsilon} \cup \overline{\Upsilon}, \forall (i,j) \in \mathcal{E}\}.
\ee 
%--------------------------------------------------------------------------
\begin{proposition}\label{co:clus}
Consider the discrete-time Markov chain in \eqref{eq:rel} representing the relative phase dynamics of $n$ interconnected oscillators over a connected and undirected line network under Assumptions \ref{as1} and \ref{ass2}. Assume that ${\boldsymbol E}[{\tilde\omega}_{i,j}(\mathrm{k})]=0, \forall (i,j)$, and each of the multiplicative uncertainties obeys ${\tilde \alpha}_{i,j} \sim {\mathcal N}(\mu_{i,j}, \sigma^2)$, where $\mu_{i,j}$ is either positive or negative such that $|\mu_{i,j}|=\lambda >0$, and $\lambda \gg \sigma^2$. If 
$$\kappa >0 \quad {\text {and}}\quad \tau < \frac{\gamma}{2 \kappa \lambda \Psi_{\max}},$$
holds, then the relative phase process is stochastic phase-cohesive with respect to $U^{\Psi}_{G}(\gamma)$ in~\eqref{eq:s2} such that each relative phase whose uncertain coupling weight has a positive (negative) mean value is recurrent to the in-phase arc $\underline\Upsilon$ in \eqref{eq:a1} (anti-phase arc $\Upsilon$ in \eqref{eq:a3}).   
\hfill \QED
\end{proposition}
%================================================================================================
\subsection{Relaxation of the odd coupling function}
Definition of $\underline{\Upsilon}$ in \eqref{eq:a1} assumes that $0 \in \underline{\Upsilon}$ and $\Psi(0) =0$. However, the latter condition can be relaxed such that the coupling function takes a zero value at a non-zero arc $\gamma_c: 0< \gamma_c <\pi$. This allows considering cases where the coupling is not an odd function on the entire interval $[-\pi,\pi]$. 

\begin{manualassumption}{\ref{as1}'}\label{as3}
Function $\Psi_{r}(\cdot)$~is $2 \pi$-periodic, and continuously differentiable. There exists an arc $\Upsilon \subset (0,\pi)$ such that $\forall \xi , z \in [0,\pi]$, if $\xi \not\in \Upsilon, z \in \Upsilon$, then $|\Psi_{r}(\xi)| \leq |\Psi_{r}(z)|$. In the interval $(0,\pi)$, there exists $\gamma_c \not\in \Upsilon$ such that $\gamma_c < \gamma, \forall \gamma \in \Upsilon$, and it holds that $\gamma_c$ and $\pi$ are the only roots of $\Psi_r(\cdot)$ in $[\gamma_c,\pi]$. Moreover, function $\Psi_{r}(\cdot)$ is odd on $[-\pi,-\gamma_c] \cup [\gamma_c, \pi]$, and $\forall \gamma_i \in [-\gamma_c, \gamma_c]: |\Psi_{r}(\gamma_i)| \leq \bar\Psi < |\Psi_{r}({\gamma})|$.
\end{manualassumption}

\noindent Following Assumption 1', we construct the two following sets:
\be\label{eq:a10} 
{\underline{\Upsilon}}^1=[0,\gamma_c]: \{\forall \gamma_i \in {\underline{\Upsilon}}^1: 0 \leq |\Psi_{r}({\gamma}_i)| \leq \bar{\Psi}\},\hspace{3mm}
\ee
\be\label{eq:a20} 
{\underline{\Upsilon}}^2=[\gamma_c, \gamma]: \{\forall \gamma_i \in {\underline{\Upsilon}}^2: 0 \leq |\Psi_{r}({\gamma}_i)| \leq |\Psi_{r}({\gamma})|\}.
\ee
Accordingly, we modify the definition of $S^{\Psi}_{G}(\gamma)$ as follows:
\be\begin{aligned}\label{eq:s10}
S^{\Psi_{r}}_{G}(\gamma)=\{\theta_i,\theta_j \in \S^1: |\theta_{i,j}(\mathrm{k})| \in {\underline{\Upsilon}^1} \cup {\underline{\Upsilon}}^2, \forall (i,j) \in \mathcal{E}\}. 
\end{aligned}\ee
%===============================================================================================
\begin{proposition}\label{co:sh}
Consider the discrete-time Markov chain in \eqref{eq:rel} over a tree network topology and under Assumptions \ref{as3} and \ref{ass2}. Assume that ${\boldsymbol E}[{\underaccent{\sim}{\boldsymbol \alpha}}]$ is positive definite, and $\lambda_{\min}({\boldsymbol E}[{\underaccent{\sim}{\boldsymbol \alpha}}])=\mu_m \gg \sigma^2$ holds. If ${\hat\lambda}=\mu_{m} |\Psi_{r}(\gamma)|- \mu_{M} \bar{\Psi} \sqrt{(m-1)}>0$, then the relative phase process is stochastic phase-cohesive, with the absorbing set $S^{\Psi_{r}}_{G}(\gamma)$ in \eqref{eq:s10} provided that 
\begin{subequations}\label{eq:kappa-sh}\be
%\kappa > \frac{{\boldsymbol E}_{\max} [|\Delta{\tilde{\omega}}|]}{\hat\lambda \ \lambda_{\min}(L_e)},\ee\\[0.25mm]
\kappa > \frac{\big({\hat\lambda}_1+{\hat\lambda}_2 (m-1)\big) {\boldsymbol E}_{\max} [|\Delta{\tilde{\omega}}|]}{{\hat\lambda} \lambda_{\min}(L_{e}) \big({\hat\lambda}_1+{\hat\lambda}_2 \sqrt{(m-1)}\big)},\ee\\[0.25mm]
\be \tau < \frac{\gamma}{\kappa \Psi_{\max} \lambda_{\max}(L_e)\  \mu_{M} + {\boldsymbol E}_{\max} [|\Delta{\tilde{\omega}}|]},\hspace{3mm}
\ee\end{subequations}
where $\mu_m=\lambda_{\min}({\boldsymbol E}[{\underaccent{\sim}{\boldsymbol \alpha}}])$, $\mu_M=\lambda_{\max}({\boldsymbol E}[{\underaccent{\sim}{\boldsymbol \alpha}}])$, ${\boldsymbol E}_{\max} [\cdot]\!=\! {\underaccent{i,j}{\max}}{\boldsymbol E}[\cdot]$, ${\hat\lambda}_1=\mu_m |\Psi_{r}(\gamma)|$, ${\hat\lambda}_2=\mu_M \bar{\Psi}$, and $\lambda_{\max}(L_e)$ and $\lambda_{\min}(L_e)$ are respectively the largest and smallest eigenvalues of the underlying tree graph's edge Laplacian.
\hfill \QED 
\end{proposition}
%============================================================================================================================================
\begin{remark}\label{re}[Insights from stochastic stability analysis] Problem~\ref{p1} considers stochastic uncertainties, with continuous probability distributions, which can take positive, negative, or zero values at every sample time with no restriction on their amplitudes. Compared with the deterministic stability, which depends on either the upper bound \cite{franci2010phase} (for additive), or the bound, sign, and behavior (for multiplicative) of the time-evolution of disturbances \cite{lu2018stability}, our results only require finite means and bounded variances. To elaborate further, consider phase oscillators in a network where the mean-values of multiplicative uncertainties are positive, and the mean values of additive noises are zero, \ie, ${\boldsymbol E}[{\underaccent{\sim}{\boldsymbol \alpha}}]>0$ and ${\boldsymbol E}_{\max} [|\Delta{\tilde{\omega}}|]=0$. Based on Theorem \ref{p1}, for any $\kappa >0$, the in-phase set is stochastic phase-cohesive. This also includes the case of $\gamma=0$, indicating stability of the phase-locked (exact synchronization) solution, despite the fact that the samples of all uncertainties can take any value at any time step. In comparison, for the deterministic Kuramoto model with the all-to-all topology, local input-to-state stability of the phase-locked solution given bounded additive perturbations, has been proved \cite{franci2010phase}.
We also notice that despite the effects of multiplicative uncertainties on the network topology at each sample time, our stochastic analysis does not require considering various possible time-varying typologies. Next section studies a specific model of multiplicative uncertainties, \ie, random connections. We show that stochastic phase-cohesiveness depends on the probability of interconnections of oscillators, comparable with results to the stochastic linear consensus problem \cite{hatano2005agreement}, and different from the deterministic setting, where stability depends on the time-varying graphs.
\end{remark}
%%%%%%%%%%%%%%%%%%%%%%%%%%%%%%%%%%%%%%%%%%%%%%%%%%%%%%%%%%%%%%%%%%%%%%%%%%%%%%%%%%%%%%%%%%%%%%%%%%%%%%%%%%%%%%%%%%%%%%%%%%%%%%%%%%%%%%%%%%%%%%%%
\section{Oscillators in a random network}\label{sec:rand}
In this section, we investigate Problem \ref{p2} in order to characterize conditions under which the stochastic phase-cohesiveness is achieved for the coupled oscillators over an Erd\"{o}s-R\'{e}nyi random network. Considering each oscillator's dynamics in \eqref{eq:p2}, the relative phase dynamics of oscillators $i$ and $j$ obeys
\begin{align}\label{eq:rand1}
{\theta_{i,j}}(\mathrm{k}+1)&={\theta_{i,j}}(\mathrm{k})+ \tau {\omega}_{i,j}\\ \nonumber & -\Big(\kappa\tau \sum_{\ell \in \{i,j\}} \sum_{e \in {\cal N}_\ell} {\beta}_{\ell,e}(\mathrm{k}) \Psi(\theta_{\ell,e}(\mathrm{k}))\Big),
\end{align}
where ${\theta_{i,j}}(\mathrm{k})$ and $\omega_{i,j} \in \R$ represent the relative phase and relative exogenous frequency of two oscillators $i$ and $j$, respectively. The random variable ${\beta}_{\ell,e}(\mathrm{k})$ obeys the Bernoulli distribution, \ie,
\begin{equation}\label{eq:bd}
{\beta}_{\ell,e}(\mathrm{k})=
\begin{cases}
1\ \text{with probability \ $p$}\\
0\  \text{with probability \ $1-p$}.
\end{cases}
\end{equation}
Considering a probable absence of each edge in the network, a graph with the incidence matrix $B^{\mathrm{k}}$ presents the topology of the network at time $\mathrm{k}$. We denote the set of incidence matrices associated with random graphs, with $n$ nodes and probability of link failure of $1-p$, by ${\mathbf B}(n,p)$. Thus, $B^{\mathrm{k}} \in {\mathbf B}(n,p)$. To stay consistent with the problem formulation of the previous section and without loss of generality, we assume the existence of a maximal graph for the given problem, for instance a complete graph, denoted by $G$ whose incidence matrix is $B$. We then model the randomness with the term $B {\boldsymbol \beta}({\mathrm{k}})$ where ${\boldsymbol \beta}_{m \times m}(\mathrm{k})$ is a diagonal matrix capturing the random interconnections.

Thus, the relative phase vector for the whole network, $\Theta(\mathrm{k})$ in \eqref{eq:T}, is equal to $B^{\top} {\boldsymbol\theta}(\mathrm{k})$. The compact form of the relative phase dynamics follows
\begin{equation}\label{eq:relran}
B^\top {\boldsymbol\theta}(\mathrm{k}+1)= B^\top \Big(\boldsymbol\theta(\mathrm{k}) + \tau {\boldsymbol \omega}(\mathrm{k})- \tau \kappa B {\boldsymbol \beta}(\mathrm{k}) \Psi(B^\top \boldsymbol\theta(\mathrm{k}))\Big), 
\end{equation}
where ${\boldsymbol \beta}(\mathrm{k})$ is a diagonal matrix whose diagonal elements obey~\eqref{eq:bd} and
\begin{equation*}{\boldsymbol \omega}_{n \times 1}(\mathrm{k})= \big({\omega}_{1}(\mathrm{k}),\ldots,{\omega}_{n}(\mathrm{k})\big)^{\top}.\end{equation*}

The above model allows us to study the stochastic phase-cohesiveness of the network by using the developed setting in the previous sections. Different from Section \ref{sec:noise}, this section focuses on the case of constant exogenous frequencies. We now discuss the chain properties and show that $S^{\Psi}_{G}(\gamma)$ in~\eqref{eq:a1} is the stable set for this model as well. 
%----------------------------------------------------------------------------------------------------------------------------------
\begin{lemma}\label{pr3}
The relative phase stochastic process in \eqref{eq:rand1} is a $\psi$-irreducible Markov chain on a countable space.
\end{lemma}
%----------------------------------------------------------------------------------------------------------------------------------
\vspace{2mm}
The following result presents a counterpart of Theorem \ref{th2} for the case of random networks. 
\begin{theorem}\label{th3}
Consider the discrete-time Markov chain in \eqref{eq:relran} under Assumption \ref{as1}. Assume that $\omega_{i,j} \neq 0, \forall i,j$. If $p$ is strictly positive, the relative phase process is  stochastic phase-cohesive with respect to $S^{\Psi}_{G}(\gamma)$ in \eqref{eq:s1} if the following conditions hold
\begin{subequations} \be \kappa > \frac{|\Delta_{\max}{\omega}|}{|\Psi(\gamma)|\ p\ \lambda_{\min}(L_{e}({G}_{\tau}))},\ee
\be\begin{aligned} \tau < \frac{\gamma}{\kappa\ \Psi_{\max}\ \lambda_{\max}(L_e)},\hspace{3mm}% + |\Delta_{\max}{\omega}|},\hspace{3mm}
\end{aligned}\ee\label{eq:kappa-3}\end{subequations}
where $|\Delta_{\max}{\omega}|= {\underaccent{i,j}{\max}} |\omega_i- \omega_j|$, $\lambda_{\max}(L_e)$ is the largest eigenvalue of $L(G)$, and $\lambda_{\min}(L_{e}({G}_{\tau}))$ is the minimum among the smallest eigenvalues of the spanning trees of $G$.
\hfill \QED
\end{theorem} 
%-------------------------------------------------------------------------------------------------------------------------------- 
We now assume identical exogenous frequencies and derive conditions under which the chain achieves phase-synchronization (or phase-locking). With this, we show the applicability of our definition of stochastic phase-cohesiveness for a stronger notion of synchronization, \ie, phase-locking.
\begin{corollary}\label{th4}
Consider the discrete-time Markov chain in \eqref{eq:relran} under Assumption \ref{as1} and with identical exogenous frequencies. Assuming a sufficiently small $\tau$, the relative phase Markov chain is stochastic phase-cohesive with respect to the origin provided that $p>0$ and $\kappa >0$.
\hfill \QED
\end{corollary}

\begin{remark}\label{remrand}
It is worth noting the difference between the probability spaces for the uncertain network (Problem \ref{p1}) and the random network (Problem \ref{p2}). In Problem \ref{p1}, the relative phases at time $\mathrm{k}$ evolve in a general space which is generated by independent stochastic processes governing their corresponding exogenous frequencies and couplings. In Problem \ref{p2}, however, the randomness only affects the interconnection topology. As a  result, the relative phases at time $\mathrm{k}$ will transit to a new state within a countable set of states generated by independent Bernoulli processes determining the interconnection topology. 
\end{remark}
%================================================================================================================================
\section{Simulation results}\label{sec:sim}
This section presents numerical simulations to validate our theoretical results on the stochastic phase-cohesive behavior of interconnected oscillators for both uncertain (Section \ref{sec:noise}) and random (Section \ref{sec:rand}) networks. We assume that the coupling law obeys $\Psi (\theta)= \sin(\theta)+ 0.3 \sin (3 \theta)$ which meets Assumption \ref{as1}. We also define the arcs $\underline{\Upsilon}=[0,\frac{\pi}{8}], \overline{\Upsilon}=[\frac{\pi}{1.14},\pi]$. Figure \ref{fig:function} shows the plot of function $\Psi(\theta)$ as well as $\Psi_{r}(\theta)$. The latter, $\Psi_{r}(\theta)= 1.5 \sin(1.1 \theta)-0.7 \cos(3.3 \theta- 0.4 \pi)$, is not an odd function on the entire interval $[-\pi,\pi]$. This function is used in simulations designed for the verification of Proposition~\ref{co:sh}.
\begin{figure}[h!]
    \centering
    \includegraphics[scale=0.65]{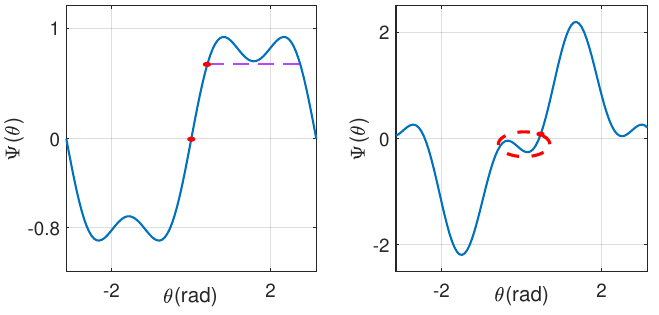}
    \caption{The plot of $\Psi(\theta)$ (left) and $\Psi_{r} (\theta)$ (right). The 'x' axis ranges over $[-\pi,\pi]$. The red dots on $\Psi(\theta)$ show the minimum and maximum possible values for $\gamma$. The red circle on $\Psi_r(\theta)$ shows where the interval in which function is not odd, and the red dot shows $\gamma_c$.}
    \label{fig:function}
\end{figure}
\begin{figure}[h!]
\centering
\includegraphics[scale=0.68]{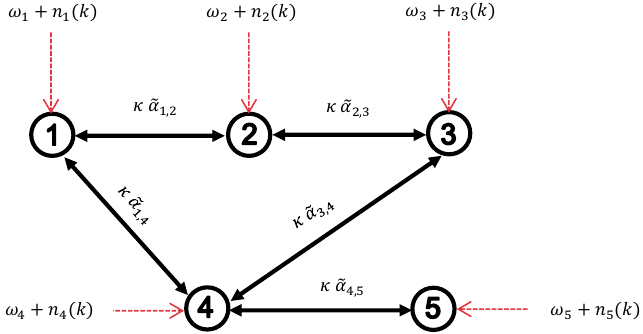}
\caption{Network of 5 oscillators subject to multiplicative and additive uncertainties. Dashed arrows illustrate assignment of exogenous frequencies.}
\label{fig:net}
\end{figure}
\subsection{Uncertain network}
Figure \ref{fig:net} shows a network composed of five oscillators. The constant components of the exogenous frequencies are set to $\omega_1=1, \omega_2=2, \omega_3=3, \omega_4=4, \omega_5=5$. The initial conditions for the oscillators are set to $\boldsymbol \theta(0)=[\frac{\pi}{4},\frac{\pi}{8},\frac{-\pi}{8},\frac{-\pi}{5},\frac{\pi}{5}]$. 
The multiplicative and additive stochastic uncertainties are modeled by Gaussian random variables and reported by ${\tilde\alpha}_{\ell,e}^{(m,v)}=(\text{Mean value}, \text{Variance})$ and ${\varpi}_i^{(m,v)}=(\text{Mean value}, \text{Variance})$. For the first experiment, the mean values and variances are set based on the following table. Calculating the bound for $\kappa$ based on Theorem \ref{th1}, we obtain $\kappa > 39.8$ by replacing ${\boldsymbol E}_{\max} [|\Delta{\tilde{\omega}}|]= 3$, $\mu_{m}-\sqrt{\frac{2 \sigma^2}{\pi}}s=0.3$, $|\Psi(\gamma)|=0.66$, and $\lambda_{\min}(L_{e}({G}_{\tau}))=0.38$. The latter is the minimum eigenvalue of the network's spanning tree which is a line graph obtained by removing the edge $(3,4)$. We set $\kappa= 40$ and $\tau=0.001$ meeting the requirements of Theorem~\ref{th1}. 
\begin{center}
 \begin{tabular}{| l | l |} 
 \hline
 ${\tilde\alpha}_{\ell,e}^{(m,v)}$ & ${\varpi}_i^{(m,v)}$\\ [0.5ex] 
 \hline\hline
 ${\tilde\alpha}_{1,2}^{(m,v)}=(1,0.5)$ & ${\varpi}_{1}^{(m,v)}=(4,1)$\\ 
 \hline
 ${\tilde\alpha}_{2,3}^{(m,v)}=(3,0.5)$ & ${\varpi}_{2}^{(m,v)}=(2,2)$\\
 \hline
 ${\tilde\alpha}_{3,4}^{(m,v)}=(0.85,0.5)$ & ${\varpi}_{3}^{(m,v)}=(0,1)$\\
 \hline
 ${\tilde\alpha}_{1,4}^{(m,v)}=(1.5,0.5)$& ${\varpi}_{4}^{(m,v)}=(1,3)$\\
 \hline
 ${\tilde\alpha}_{4,5}^{(m,v)}=(2,0.5)$ & ${\varpi}_{5}^{(m,v)}=(-2,1.5)$\\ [1ex] 
 \hline
\end{tabular}
\end{center}
The time-evolution of the oscillators' phases, the relative phases, and the maximum relative phase are shown in Figure~\ref{fig:un1}. As shown the relative phases are confined in the desired set.
\begin{figure}[h!]
    \centering
    \includegraphics[scale=0.51]{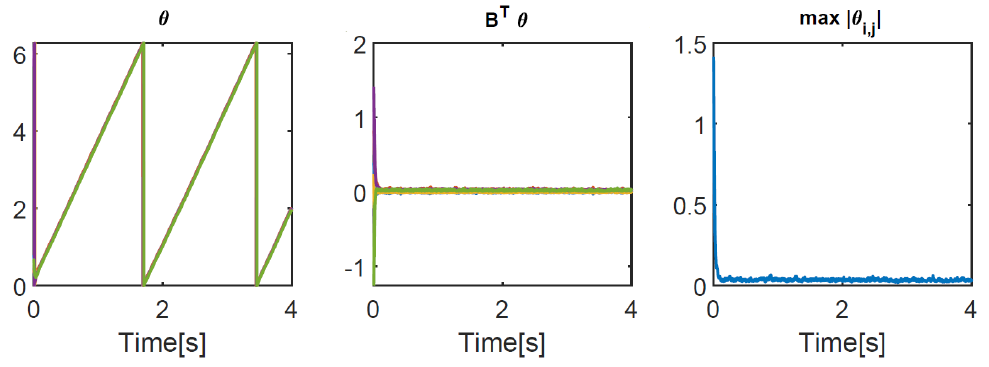}
    \caption{Phases, relative phases and the maximum relative phase: uncertain couplings (positive mean), and exogenous frequencies (Theorem~\ref{th1}.1).}
    \label{fig:un1}
\end{figure} 

In order to examine the effects of multiplicative uncertainties with negative mean values, we keep the settings of the first experiment but replace the mean values of the multiplicative randomness with negative values. The time evolution of relative phases where all multiplicative mean values are set to negative ones are shown in Figure~\ref{fig:un2}. The results confirm the stochastic phase-cohesiveness w.r.t the anti-phase set.
\begin{figure}[h!]
    \centering
    \includegraphics[scale=0.44]{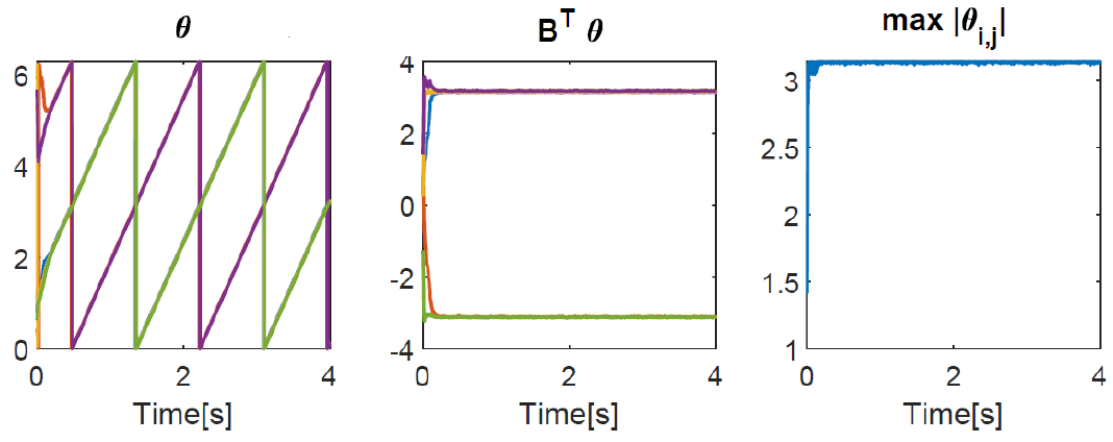}
 \caption{Phases, relative phases and the maximum relative phase: couplings (negative mean values), and uncertain exogenous frequencies (Theorem~\ref{th1}.2).}
    \label{fig:un2}
\end{figure}

To verify Propositions \ref{co:clus} and \ref{co:sh}, we consider a line graph obtained by removing the edge $(3,4)$ of the graph shown in Figure~\ref{fig:net}. We examine the results of Proposition~\ref{co:clus} by setting the size of all mean values of the multiplicative uncertainties equal to one. We set ${\tilde\alpha}_{1,2}^{(m,v)}=(-1,0.5);{\tilde\alpha}_{2,3}^{(m,v)}=(-1,0.5)$ and keep the rest of the mean values positive. We also set ${\boldsymbol E}[{\tilde\omega}_{i}]=0, \forall i$. The coupling coefficient $\kappa>0$ is set to $\kappa=2$. As shown in Figure~\ref{fig:un22}, the relative phases form two clusters. 
\begin{figure}[h!]
    \centering
    \includegraphics[scale=0.44]{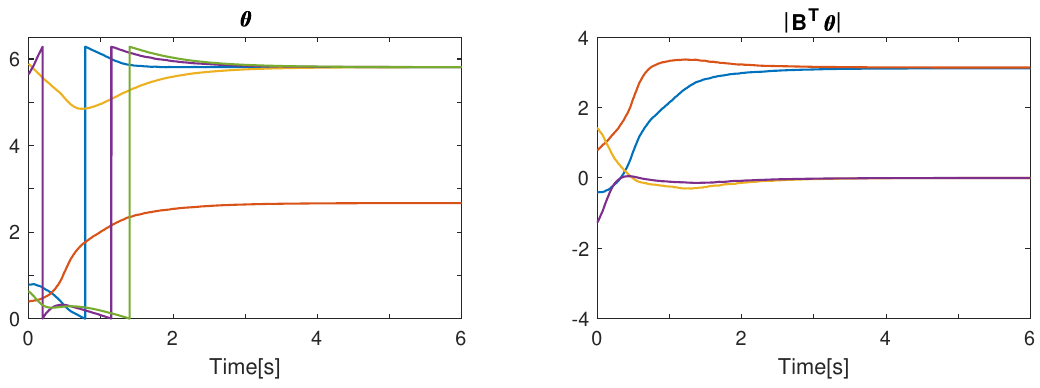}
 \caption{Phases and relative phases over a line network with multiplicative uncertainties with mixed positive and negative mean values (Proposition~\ref{co:clus}).}
    \label{fig:un22}
\end{figure}

To verify the result of Proposition \ref{co:sh}, we use the setting of the first experiment but with a non-odd coupling function $\Psi_{r}(\theta)$. We set $\Psi_{r}(\gamma=0.4\pi)= 2$ and $\bar\Psi=0.2$. We have $\mu_m=1; \mu_M=3; m=4$, which gives $\hat\lambda=0.98$ and $\kappa > 10$.  We set $\kappa=10$. As shown in Figure~\ref{fig:un3}, the relative phases are bounded and the network behavior follows the result of Proposition \ref{co:sh}.
\begin{figure}[h!]
    \centering
    \includegraphics[scale=0.44]{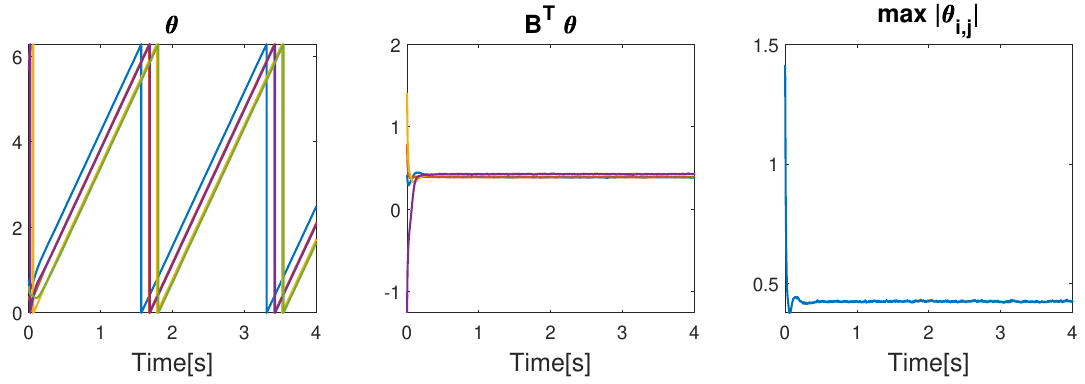}
 \caption{Phases, relative phases, the maximum relative phase: Coupled by $\Psi_{r}(\theta)$, uncertain exogenous frequencies, and couplings with positive mean values (Proposition~\ref{co:sh}).}
    \label{fig:un3}
\end{figure}

\subsection{Random network}
We now present the simulation results for a random network with five oscillators. We use the coupling function $\Psi(\cdot)$, $\underline{\Upsilon}=[0,\frac{\pi}{8}]$, and the initial conditions of the oscillators similar to the first experiment of the previous section. If all links are connected the graph depicted in Figure \ref{fig:net} is obtained, hence the maximal graph. The non-zero and non-identical constant exogenous frequencies equal to $\omega_1=1, \omega_2=2, \omega_3=3, \omega_4=4, \omega_5=5$. Sampling time is set to $0.01 s$ to elaborate the effects of the randomness. We first assume that $p=0.8$. Using the condition in \eqref{eq:kappa-3}, we calculate $\kappa > 15$ and set $\kappa=19$. The evolution of the oscillators' phases, the relative phases and the maximum relative phase over time are shown in Figure \ref{fig:ran1}-(a). As shown, the relative phases are confined in the set $S^{\Psi}_{G}(\gamma=\frac{\pi}{8})$.  

We then decrease the connectivity probability to $0.3$. The results, reported in Figure \ref{fig:ran1}-(b), show that the relative phases are not bounded. We update the value of $\kappa$ with respect to the decrease in the connection probability to $\kappa=30$. The plots in Figure \ref{fig:ran1}-(c) show that the maximum relative phase is now within the desired set.
\begin{figure}[h!]
    \centering
    \includegraphics[scale=0.65]{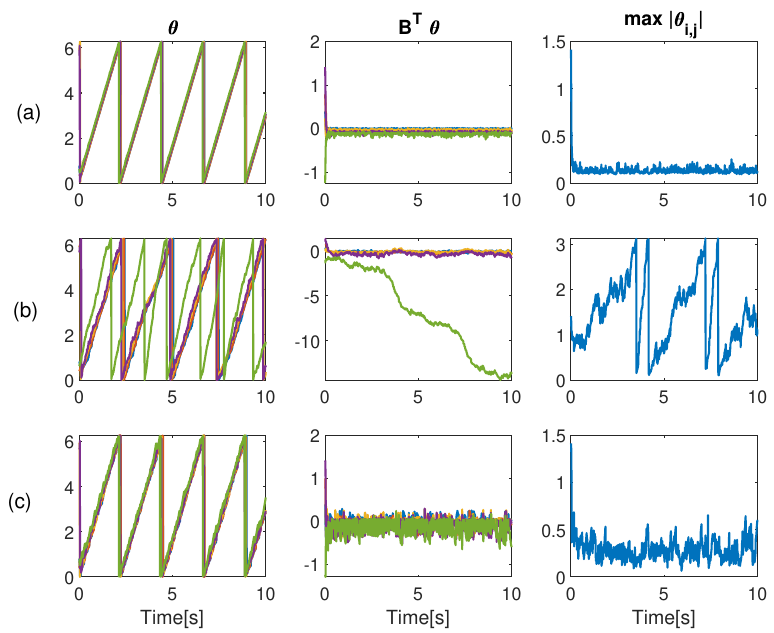}
    \caption{Phases, relative phases and the maximum relative phase over random network: non-identical exogenous frequencies: (a) $p=0.8, \kappa=19$, (b) $p=0.3, \kappa=12$, (c) $p=0.3, \kappa=30$ (Theorem~\ref{th3}).}
    \label{fig:ran1}
\end{figure}
Next, we assume all exogenous frequencies are set to one. We set $\kappa=0.5 >0$, and consider two probabilities of connection: $p=0.8$ and $p=0.1$. Plots in Figure \ref{fig:ran2} show the time-evolution of the phases, the relative phases and the maximum relative phase over time for these two cases. As shown all oscillators' relative phases converge to zero and the rate of convergence is proportional to the probability of connection.
\begin{figure}[h!]
    \centering
    \includegraphics[scale=0.65]{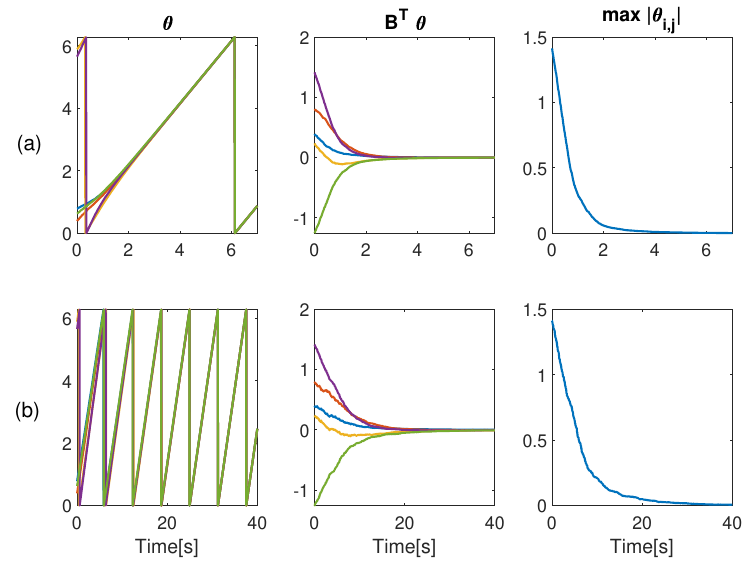}
    \caption{Phases, relative phases and the maximum relative phase over random network: identical exogenous frequencies, $\kappa=0.5$: (a) $p=0.8$, (b) $p=0.1$.}
    \label{fig:ran2}
\end{figure}
%================================================================================================================================
\section{Conclusions}\label{sec:cl}
This article studied stochastic relative phase stability for a class of discrete-time coupled oscillators. The two notions of stochastic phase-cohesiveness and ultimate stochastic phase-cohesiveness were introduced. Stochastic phase-cohesiveness of oscillators, with a general class of $2\pi$-periodic, and odd coupling functions, with respect to the two, in-phase and anti-phase, sets were studied. We investigated undirected networks subject to both multiplicative and additive stochastic uncertainties. We proved stochastic phase-cohesiveness with respect to the in-phase set when the mean values of all multiplicative uncertainties were positive, and with respect to the anti-phase set for the case of negative mean values. In addition, we have discussed the relaxation of the odd property of the coupling function by allowing this function to be non-odd on a subset of its domain. Moreover, we proved a clustering behavior for a network with an underlying line topology subject to mixed negative and positive mean values for the multiplicative uncertainties and zero mean value for the additive uncertainties. Further, the stochastic phase-cohesiveness of oscillators with constant exogenous frequencies in an Erd\"{o}s-R\'{e}nyi random network was studied. Sufficient conditions for achieving both stochastic phase-cohesive and phase-locked solutions were derived. It was proved that oscillators with equal exogenous frequencies in a random network with any positive possibility of connection will achieve phase-locking. Our results emphasize the importance of the coupling function in synchronization, discuss the stabilizing effects of the additive stochastic uncertainties, and the effects of multiplicative uncertainties in achieving stochastic phase-cohesiveness. 
%================================================================================================================================
\appendices
%-------------------------------------------------------------------------------------------------------
\section{Proof of Lemma~\ref{mc1}}\label{amc1}
\begin{proof}
The relative phase process $\Theta(\mathrm{k})=B^\top \boldsymbol\theta(\mathrm{k})$ with i.i.d. stochastic variables, generated by continuous distributions, satisfies all properties in Lemma 1 according to the definition of Markov chains, and Definitions \ref{def:markov}-\ref{def:irreducibility}.
\end{proof}
%-------------------------------------------------------------------------------------------------------
\section{Proof of Lemma~\ref{pr2}}\label{apr2}
\begin{proof}
Based on Proposition 6.2.8. in \cite{meyn2012markov}, the satisfaction of the Feller property (Definition \ref{def:feller1}) together with the non-emptiness of the support of the irreducibility measure leads to the conclusion that every compact set in the state space is also petite. The Feller property of the Markov chain \eqref{eq:rel} can be readily concluded according to the Proposition 6.1.2. in \cite{meyn2012markov}, since the right-side of the Markov chain  \eqref{eq:rel} is a continuous function in $\boldsymbol \theta(\mathrm{k})$ for each fixed pair of i.i.d. realizations $(\tilde{\boldsymbol \omega}(\mathrm{k}),\underaccent{\sim}{\boldsymbol \alpha}(\mathrm{k}))$. In addition, the %additive and multiplicative 
uncertainties are capable of forcing transitions from any subset of the $\sigma$-algebra to an open petite set in the state space, hence, there exists no set in the $\sigma$-algebra in which if the Markov state enters, it always remains there with the absolute probability of one. This concludes that the support of the irreducibility measure has a non-empty interior which completes the proof. 
\end{proof}
%--------------------------------------------------------------------------------------
\section{Lemma 3: Statement and Proof}\label{ax}
\begin{lemma}\label{lem1}
Consider the random vector 
$$y({\mathrm{k}})= B^\top \boldsymbol\theta(\mathrm{k})-\tau \kappa B^\top B {\underaccent{\sim}{\boldsymbol \alpha}}(\mathrm{k}) \Psi(B^\top \boldsymbol\theta(\mathrm{k})),$$ where $y({\mathrm{k}}) \in \R^{m \times 1}$, and ${\underaccent{\sim}{\boldsymbol \alpha}}(\mathrm{k})$ is a diagonal matrix whose elements are i.i.d Gaussian random variables as in Assumption~\ref{ass2}. Then, the following inequality holds:
$${\boldsymbol E}[|y({\mathrm{k}})| \big| \boldsymbol\theta(\mathrm{k})] \leq |{\boldsymbol E} [y({\mathrm{k}}) \big|\boldsymbol\theta(\mathrm{k})]|+ \kappa \tau \sqrt\frac{2 \sigma^{2}}{\pi} |B^{\top} B \Psi(B^\top \boldsymbol\theta(\mathrm{k}))|.$$
\end{lemma}
\begin{proof}
The proof is based on an application of the folded normal distribution~\cite{leone1961folded}, for calculation of the expectation of the absolute value of a random variable $Z \sim N(\mu,\sigma^2)$. We have,
$${\boldsymbol E}[|Z|]= \sqrt\frac{2 \sigma^2}{\pi} \exp(\frac{-\mu^2}{2 \sigma^2})+ \mu\ erf(\frac{\mu}{\sqrt{2\sigma^2}}),$$
where {\em erf}, the error function~\cite{leone1961folded}, is an odd function such that $|erf(.)| \leq 1$. As a result, we can write,
\be\label{eq:bond} {\boldsymbol E}[|Z|]\leq \sqrt\frac{2 \sigma^2}{\pi}+ |\mu|.\ee
Now, each element of $y({\mathrm{k}})=[y_1({\mathrm{k}}) \ldots y_m({\mathrm{k}})]^\top$ is a random variable obeying the normal distribution. Calculating ${\boldsymbol E}[|y({\mathrm{k}})| \big| \boldsymbol\theta(\mathrm{k})]$, requires computation of the expectation of absolute value of each element of $y({\mathrm{k}})$, \ie, ${\boldsymbol E}[|y_i({\mathrm{k}})| \big| \boldsymbol\theta(\mathrm{k})]$. Denote $B^\top B \Psi(B^\top \boldsymbol\theta(\mathrm{k}))$ by $x=[x_1 \ldots x_m]^\top$. We have,
$${\boldsymbol E}[|y_i({\mathrm{k}})| \big| \boldsymbol\theta(\mathrm{k})] \leq |{\boldsymbol E}[y_i({\mathrm{k}})\big| \boldsymbol\theta(\mathrm{k})]| + \sqrt\frac{2 \sigma^2 \kappa^2 \tau^2 x_i^2}{\pi}.$$
Th proof is completed by stacking all $y_i$ into vector $y({\mathrm{k}})$, and all $x_i$ into $x$.
\end{proof}
%---------------------------------------------------------------------------------------------------------------------------------------
\section{Proof of Theorem~\ref{th1}}\label{ath1}
\begin{proof}
The proof is based on an application of Theorem 9.1.8 of \cite{meyn1994} in a network setting. We first assume that ${\boldsymbol E}[{\underaccent{\sim}{\boldsymbol \alpha}}]$ is positive definite and prove that the Markov chain in \eqref{eq:rel} is stochastic phase-cohesive with respect to the in-phase set $S^{\Psi}_{G}(\gamma)$. Let us assume that the chain lives in the set $$U^1= \{\theta_i, \theta_j \in \S^1: |\theta_{i,j}(\mathrm{k})| \in \underline{\Upsilon} \cup {\Upsilon}, \forall (i,j) \in \mathcal{E}\},$$ where $\underline{\Upsilon}$ and ${\Upsilon}$ are defined in \eqref{eq:a1} and \eqref{eq:a2}, respectively.\\[1mm]
Let us first assume that ${\boldsymbol E}[{\underaccent{\sim}{\boldsymbol \alpha}}]$ is positive definite and define the positive and radially unbounded function
\be\label{eq:lya}V(\Theta({\mathrm{k}}))= |\Psi(\gamma)|\!\!\!\!\sum\limits_{|\theta_{i,j}(\mathrm{k})| \in {\Upsilon}}\!\!\!\!|\theta_{i,j}(\mathrm{k})| + \Psi^{o}\!\!\!\!\sum\limits_{|\theta_{i,j}(\mathrm{k})| \in \underline{\Upsilon}}\!\!\!\!|\theta_{i,j}(\mathrm{k})|,\ee
where $\Theta({\mathrm{k}})=B^\top \boldsymbol\theta({\mathrm{k}})$, $\Psi^{o}= \Psi(0^+)$ and $V:\Pi \rightarrow \R^+$, with $m$ the total number of edges of the underlying graph. Notice that in the view of Lemma~\ref{pr2}, the level sets of $V$ are petite. Let us assume that $\left\{\exists (\ell,q)\ \text{s.t.}\quad |\theta_{\ell}(\mathrm{k})-\theta_{q}(\mathrm{k})| \in {\Upsilon}\right\}$. We now derive conditions under which the one-step drift of $V$, \ie,  
\begin{equation}\begin{aligned}\label{eq:dv11}
\Delta V(\Theta)={\boldsymbol E} [V(\Theta({\mathrm{k}+1})) \big|\Theta({\mathrm{k}})]- V(\Theta({\mathrm{k}})),
\end{aligned}\end{equation}
is negative. Notice that the oscillators' phases at time $\mathrm{k}$, \ie, $\boldsymbol \theta(\mathrm{k})$ is known. %In what follows, we argue that $V(\Theta)$ decreases at one time step. As a result, we conclude that $V(\Theta({\mathrm{k}}+1))\leq V(\Theta({\mathrm{k}}))$, which indicates a negative drift.

Let us denote the coefficients of $V(\Theta({\mathrm{k}}))$ by $C \in \R^{m \times 1}$. Then, we can write $V(\Theta({\mathrm{k}}))= C^{\top} |B^\top \boldsymbol\theta(\mathrm{k})|$. From the relative phase dynamics in \eqref{eq:rel}, we have
\begin{equation}\begin{aligned}\label{eq:dvc}
&{\boldsymbol E} [V(\Theta({\mathrm{k}+1})) \big|\Theta({\mathrm{k}})]= C^{\top} {\boldsymbol E} [|B^\top \boldsymbol\theta(\mathrm{k}+1)| \big|\boldsymbol\theta({\mathrm{k}})]=\\
&C^{\top} {\boldsymbol E} [|B^\top \boldsymbol\theta(\mathrm{k})-\tau \kappa B^\top B {\underaccent{\sim}{\boldsymbol \alpha}}(\mathrm{k}) \Psi(B^\top \boldsymbol\theta(\mathrm{k}))+ \tau B^{\top} {\tilde{\boldsymbol \omega}}(\mathrm{k})|] \leq \\ &C^{\top} \big({\boldsymbol E}[|\underbrace{B^\top \boldsymbol\theta(\mathrm{k})-\tau \kappa B^\top B {\underaccent{\sim}{\boldsymbol \alpha}}(\mathrm{k}) \Psi(B^\top \boldsymbol\theta(\mathrm{k}))}_{I_1}| + \underbrace{\tau |B^{\top} {\tilde{\boldsymbol \omega}}(\mathrm{k})|}_{I_2}\big),
\end{aligned}\end{equation}
where $I_1=(I_1^1, \ldots, I_1^m)^\top \in \R^{1 \times m}$ and $I_2 \in \R^{m \times 1}$. According to Lemma~\ref{lem1}, we have 
\begin{equation}
{\boldsymbol E}[|I_1| \big|\boldsymbol\theta(\mathrm{k})] \leq |{\boldsymbol E} [I_1 \big|\boldsymbol\theta(\mathrm{k})]|+ \kappa \tau \sqrt\frac{2 \sigma^{2}}{\pi} |B^{\top} B \Psi(B^\top \boldsymbol\theta(\mathrm{k}))|.
\end{equation}
For the clarity of presentation, we denote $V(\Theta({\mathrm{k}}))$ and $\Psi^\top (B^\top \boldsymbol\theta(\mathrm{k}))$ by $V({\mathrm{k}})$ and $\Psi^\top(\mathrm{k})$ respectively. Also, define $r=\kappa \tau \sqrt\frac{2 \sigma^{2}}{\pi}$. Therefore from \eqref{eq:dv11}, we can write
\begin{equation}\begin{aligned}\label{eq:I}
\Delta V (\Theta) \leq C^{\top} \big(|{\boldsymbol E} [I_1]| + r |B^{\top} B \Psi^\top(\mathrm{k})|+ {\boldsymbol E} [I_2]\big) - V({\mathrm{k}}).
\end{aligned}\end{equation} 
Based on the definition of $\Psi^o$, $\Psi^o \leq |\Psi(\theta_{i,j}({\mathrm{k}}))|$ for edges with a non-zero $\Psi(\cdot)$ at each time $\mathrm{k}$. Also, since $\Psi(\cdot)$ is an odd function , and $\Psi(0)=0$, if $\theta_{i,j}({\mathrm{k}})=0$, we can replace $\Psi^{o}$ with $\Psi(\theta_{i,j}({\mathrm{k}}))=0$. Thus, $C \leq |\Psi (B^\top \boldsymbol\theta(\mathrm{k}))|$ holds element-wise, and we can write
\begin{equation}\label{eq:dv12} \Delta V (\Theta) \leq |\Psi^\top(\mathrm{k})| \big(|{\boldsymbol E} [I_1]| + r |B^{\top} B \Psi^\top(\mathrm{k})|+ {\boldsymbol E} [I_2]\big) - V({\mathrm{k}}).\end{equation}
Considering the definition of $V(\Theta({\mathrm{k}}))$ in \eqref{eq:lya}, at each time $\mathrm{k}$ it holds that $V_m(\Theta({\mathrm{k}})) \leq V(\Theta({\mathrm{k}})) \leq V_M(\Theta({\mathrm{k}}))$, where $V_m$ represents the case in which at time $\mathrm{k}$ only one relative phase belongs to $\Upsilon$, and all other relative phases are equal to zero, and $V_M$ is the case where all relative phases at time $\mathrm{k}$ belong to $\Upsilon$. In what follows, we continue the proof by obtaining conditions under which both $V_m$ and $V_M$ decrease at one time-step. Thus, $V$ necessarily decreases.\\[1mm]
{\bf{Proof of $\Delta V_m(\Theta) \leq 0$:}} Given the conditions of this case, only one relative phase at time $\mathrm{k}$ belongs to $\Upsilon$, and all other relative phases are equal to zero, \ie, $|(\theta_{\ell,q} (\mathrm{k}))|\geq \gamma$ and $|\theta_{p,s} (\mathrm{k})|= 0, \forall (p,s) \neq (\ell,q)$. From \eqref{eq:dv12}, all elements of the vector $|\Psi^{\top}(\mathrm{k})||{\boldsymbol E} [I_1]|$ are zero except the element corresponding to $\theta_{\ell,q}$. Therefore, the following equality holds:
$$|\Psi^{\top}(\mathrm{k})| |{\boldsymbol E} [I_1]|= |\Psi^{\top}(\mathrm{k}) {\boldsymbol E} [I_1]|.$$
From \eqref{eq:dvc} and \eqref{eq:dv12}, the on-step drift of $V_m$ obeys
\be\begin{aligned}\label{eq:dv2}
&\Delta V_m (\Theta) \leq \!\!\ |\!\ \! \underbrace{\Psi^\top (\mathrm{k}) B^\top \boldsymbol\theta(\mathrm{k})}_{a} -\underbrace{\tau \kappa\ \Psi^\top (\mathrm{k}) B^\top B\ {\boldsymbol E}[{\underaccent{\sim}{\boldsymbol \alpha}}(\mathrm{k})] \Psi(\mathrm{k})}_{b}|\\&\hspace{-1mm}+{\boldsymbol E}[\underbrace{\ \tau |\Psi^{\top}(\mathrm{k})| |B^{\top} {\tilde{\boldsymbol \omega}}(\mathrm{k})|}_{c}] + \underbrace{r |\Psi^\top(\mathrm{k})||B^{\top} B \Psi^\top(\mathrm{k})|}_{d} - V_{m}(\mathrm{k}).
\end{aligned}\ee
To have $\Delta V_m <0$, the followings should hold
\be\begin{aligned}\label{eq:2con}
&[1]&{\boldsymbol E}[a-b+c]+d < V_{m}(\mathrm{k}) \ \ \ \text{when} \ {\boldsymbol E}[a - b] > 0,\hspace{0.5mm}\\
&[2]&{\boldsymbol E}[-a+b+c]+d < V_{m}(\mathrm{k})\  \text{when} \ {\boldsymbol E}[a - b] < 0.
\end{aligned}\ee
We now discuss that the expectation of the term denoted by $b$ is positive. Assume ${\boldsymbol E}[{\underaccent{\sim}{\boldsymbol \alpha}}]$ is positive definite. Since $B^\top B \geq 0$ holds for a connected graph \cite{mesbahi2010graph}, and also $\Psi(\cdot)$ is an odd function, we obtain
\be\label{eq:I1}{\boldsymbol E}[b] \geq \kappa\ \tau \lambda_{\min}({\boldsymbol E}[{\underaccent{\sim}{\boldsymbol \alpha}}])\;\Psi^{\top}(\mathrm{k})\ B^\top B\ \Psi(\mathrm{k}).\ee 
Notice that inequality \eqref{eq:I1} is greater than zero. The reason is that the range space of $B^\top$ (\ie, ${\cal R} (B^\top)$) and null space of $B$ (\ie, ${\cal N} (B)$) are perpendicular. As a result, $B \Psi(B^\top \boldsymbol\theta(\mathrm{k}))= \mathbf{0}$ holds if and only if $\Psi(B^\top \boldsymbol\theta(\mathrm{k}))= \mathbf{0}$. Since we assumed that at time $\mathrm{k}$, there is at least one edge of the graph whose corresponding relative phase belongs to ${\Upsilon}$, $B^\top \boldsymbol\theta(\mathrm{k}) \neq 0$ holds. Thus, ${\boldsymbol E}[b] >0$. Now, consider the term $d$ in \eqref{eq:dv2}. Since only one element of $\Psi(\mathrm{k})$ is non-zero and $B^\top B \geq 0$, we can write $d=r \Psi^\top(\mathrm{k}) B^{\top} B \Psi^\top(\mathrm{k})$.

Consider the first inequality in \eqref{eq:2con}. Recall that $|(\theta_{\ell,q} (\mathrm{k}))|\geq \gamma$ and $|\theta_{p,s} (\mathrm{k})|= 0, \forall (p,s) \neq (\ell,q)$. This gives $a= \gamma \Psi (\gamma)$. Calculating the term $V_{m}(\mathrm{k})$ under the same condition, we obtain $\gamma |\Psi(\gamma)|$ which cancels out with $a$. We now proceed to characterize conditions which guarantee $\Delta V_m <0$. From \eqref{eq:2con}, two following criteria should hold:
\be\begin{aligned}\label{eq:two}
&\min\ \{{\boldsymbol E}[b]\} - d> \max\ {\boldsymbol E}[c],\\
&-\min\ \{a\}+\max\ \{{\boldsymbol E}[b] + {\boldsymbol E}[c]\} +d < \min\ \{V_{m}(\mathrm{k})\}.
\end{aligned}\ee
To obtain the lower bound of $b$, we write
$B^\top B= R^\top L_{e}({G}_{\tau}) R$ where $L_{e}({G}_{\tau})=B_{\tau}^\top B_{\tau} >0$ is the corresponding edge Laplacian of a spanning tree of graph $G$ at time $\mathrm{k}$, and $R=[I \quad T]$ (Theorem 4.3. of \cite{mesbahi2010graph}, see Section~\ref{sec:pre}). Thus, 
\be\begin{aligned}\label{eq:ie31}
{\boldsymbol E}[b] \geq \tau \kappa \lambda_{\min}({\boldsymbol E}[{\underaccent{\sim}{\boldsymbol \alpha}}]) \lambda_{\min}(L_{e}({\cal G}_{\tau})) \Psi^{\top}(\mathrm{k}) R^T R \Psi(\mathrm{k})>0.
\end{aligned}\ee
Under conditions $|(\theta_{\ell,q} (\mathrm{k}))|= \gamma$ and $|\theta_{p,s} (\mathrm{k})|= 0, \forall (p,s) \neq (\ell,q)$ and assuming that $(\ell,q) \in {\cal G}_{\tau}$, we have $\Psi^{\top}(\mathrm{k}) R^T R \Psi(\mathrm{k})= \Psi^2(\gamma)$.
Let $\mu_m$ denotes $\lambda_{\min}({\boldsymbol E}[{\underaccent{\sim}{\boldsymbol \alpha}}])$. Computing the first inequality in \eqref{eq:two}, gives 
\begin{equation}\label{eq:ie31-a}
(\tau \kappa \mu_m-r) \Psi^{2}(\gamma) \lambda_{\min}(L_{e}({G}_{\tau})) > \tau |\Psi(\gamma)| {\boldsymbol E}_{\max} [|\Delta{\tilde{\omega}}|].
\end{equation}
Considering the second inequality in \eqref{eq:two}, $\max \{b\}$ is obtained if $\Psi(\theta_{\ell,q}(\mathrm{k}))=\Psi_{\max}$. In this case, we have
\be\begin{aligned}\label{eq:ie31-b}
{\boldsymbol E}[b] \leq \tau \kappa \lambda_{\max}({\boldsymbol E}[{\underaccent{\sim}{\boldsymbol \alpha}}]) \lambda_{\max}(L_{e}) \Psi^{\top}(\mathrm{k}) \Psi(\mathrm{k}).
\end{aligned}\ee
Denoting $\lambda_{\max}({\boldsymbol E}[{\underaccent{\sim}{\boldsymbol \alpha}}])$ by $\mu_{M}$, we obtain $$d+{\boldsymbol E}[b] \leq (\tau \kappa \mu_M+r) \lambda_{\max}(L_{e}) \Psi^2_{\max}.$$ Also, calculating $\min\ \{a\}$, where $a$ is defined in \eqref{eq:dv2}, gives $a \geq \Psi_{\max} \gamma$. The reason is that for the edge belonging to $\Upsilon$, the minimum angle is $\gamma$, and its corresponding $\Psi(\cdot)$ value follows same as in calculation of $b$, hence $\Psi_{\max}$. We notice that based on definition of $d$ in \eqref{eq:lya}, for all edges belong to $\Upsilon$, $\min\{d\}= \gamma |\Psi(\gamma)|$.
Hence,
\begin{equation}\label{eq:ie31-c}\begin{aligned} -\gamma+\underbrace{(\tau \kappa \mu_M+r) \lambda_{\max}(L_{e}) \Psi_{\max}+ \tau {\boldsymbol E}_{\max} [|\Delta{\tilde{\omega}}|]}_{f} < \frac{|\Psi(\gamma)|}{\Psi_{\max}} \gamma,\end{aligned}\end{equation}
should hold. Since $|\Psi(\gamma)|<\Psi_{\max}$, we can replace \eqref{eq:ie31-c} with $f <\gamma$. Thus, we obtain the second condition in \eqref{eq:kappa-11}.\\[1mm]
%---------------------------------------------------------------------------------------------------------------------------------------------------------------------------
{\bf{Proof of $\Delta V_M(\Theta) \leq 0$:}} Given the conditions of this case, all relative phases at time $\mathrm{k}$ belong to $\Upsilon$, \ie, $\forall (\ell,q): |(\theta_{\ell,q} (\mathrm{k}))|\geq \gamma$.      
In fact, each element of vector ${\boldsymbol E} [I_1]$ in \eqref{eq:I} is positive if either the sign of its corresponding term in vector $\kappa \tau B^\top B\ {\boldsymbol E}[{\underaccent{\sim}{\boldsymbol \alpha}}(\mathrm{k})] \Psi(\mathrm{k})$ \eqref{eq:dvc} is negative or the sign is positive and the size is smaller than $\gamma$. Hence,
$\gamma > \kappa \tau (d_{\max}+1) \mu_M \Psi_{\max}$ should hold, where $d_{\max}$ is the maximum degree of the nodes of the underlying graph. Since $\lambda_{\max}(L_{e})\geq (d_{\max}+1)$, the bound in \eqref{eq:ie31-c} gives a smaller $\tau$, hence, it satisfies the required condition. Thus, %the following equality holds:
$|\Psi^{\top}(\mathrm{k})| |{\boldsymbol E} [I_1]|= \Psi^{\top}(\mathrm{k}) {\boldsymbol E} [I_1].$ 
We now write $\Delta V_M$ and argue similar to the case of $V_m$ which gives \eqref{eq:kappa-11}.\\[1mm]
%------------------------------------------------------------------------------------------------------------------------------------------------------
The sufficient coupling condition $\kappa$ in the above result depends on $|\Psi(\gamma)|$. For relative phases belonging to ${\Upsilon}$, $|\Psi(\gamma)|$ is the minimum of $\Psi(\gamma_i), \gamma_i \in {\Upsilon}$. That is, the obtained condition guarantees that the chain will return to $S^{\Psi}_{G}(\gamma)$ from $\Upsilon$. We notice that even if the chain initiates from $S^{\Psi}_{G}(\gamma)$, the presence of stochastic uncertainties can transfer the relative phase of each two oscillators from $S^{\Psi}_{G}(\gamma)$ to not only ${\Upsilon}$ but also $\overline{\Upsilon}$. Considering returning of the relative phases from $\overline{\Upsilon}$ to $S^{\Psi}_{G}(\gamma)$, we shall replace $|\Psi(\gamma)|$ in \eqref{eq:kappa-11} with $\Psi(\gamma_{\min})=\min \Psi(\gamma_i), \gamma_i \in \overline{\Upsilon}$. This substitution leads to $\kappa \rightarrow \infty$ (see Example $1$). To verify whether the arc set $\overline{\Upsilon}$ is absorbing, we shall study the evolution of the relative phase after exiting $\overline{\Upsilon}$. 
We can prove that if ${\boldsymbol E}[{\underaccent{\sim}{\boldsymbol \alpha}}]> 0$, the chain is transient with respect to $S^{\Psi}_{G}(\gamma_{\max})$. Recall that similar to $\underline{\Upsilon}$, the maximum value of $\Psi(\cdot)$ for the arcs in $\overline{\Upsilon}$ is $|\Psi(\gamma)|$. Define, 
$${\bar V}(\Theta({\mathrm{k}}))= C^{\top} (\pi {\mathbf 1}_{m}-|B^\top \boldsymbol\theta(\mathrm{k})|).$$
Computing the one-step drift, we obtain
\be \Delta {\bar V}(\Theta)= -\Delta V(\Theta).\ee
Therefore, assuming ${\boldsymbol E}[{\underaccent{\sim}{\boldsymbol \alpha}}]$ is positive definite, the lower bound for $\kappa$ in \eqref{eq:kappa-11}, which guarantees $\Delta V <0$, leads to $\Delta {\bar V} >0$. Based on Theorem 8.0.2 of \cite{meyn2012markov}, the chain is transient on $S^{\Psi}_{G}(\gamma_{\max})$. That is, the probability that the chain revisits $S^{\Psi}_{G}(\gamma_{\max})$ infinitely often is zero. From the analyses using both $V$ and $\bar V$, the conditions in \eqref{eq:kappa-11} guarantee that the chain is stochastic phase-cohesive with respect to $S^{\Psi}_{G}(\gamma)$.

Now, assume that ${\boldsymbol E}[{\underaccent{\sim}{\boldsymbol \alpha}}]$ is negative definite, and the underlying deterministic topology satisfies $B^{\top} \boldsymbol\theta \geq \gamma_{\max} {\mathbf 1}_{m}\in {\cal R} (B^\top)$. The latter assumption is imposed dealing with the existence of graph cycles (see \cite{mesbahi2010graph}) composed of an odd number of oscillators. In such a case, independent of the stochastic nature of our problem setting, it is not feasible to have all relative phases greater than some predefined limits, e.g. if $\gamma_{\max} >\frac{\pi}{2}$. Hence, at least one of the relative phases should be confined to the arc $\underline{\Upsilon}$ by the topological restrictions. Here, we exempt the latter case. To have $\Delta {\bar V} (\Theta) <0$, $C^\top |B^\top \boldsymbol\theta(\mathrm{k})|\geq C^\top {\boldsymbol E}[|B^\top \boldsymbol\theta(\mathrm{k}+1)|]$ should hold. Since ${\boldsymbol E}[|B^\top \boldsymbol\theta(\mathrm{k}+1)|] \geq |{\boldsymbol E}[B^\top \boldsymbol\theta(\mathrm{k}+1)]|$ element-wise, therefore we should have 
\be\label{eq:barv} C^\top |B^\top \boldsymbol\theta(\mathrm{k})|\geq C^\top |{\boldsymbol E}[B^\top \boldsymbol\theta(\mathrm{k}+1)]|.\ee
Similar to the previous case, we study two cases of one edge, $\bar V_m$, and all edges, $\bar V_M$, belonging to $\Upsilon$. Considering the case of $\bar V_{m}$, we assume $|(\theta_{\ell,q} (\mathrm{k}))| \leq \gamma_{\max}$ and $|\theta_{p,s} (\mathrm{k})|= \pi, \forall (p,s) \neq (\ell,q)$. Since $\Psi(\pi)=0$, we can replace $C$ with $|\Psi^\top (\mathrm{k})|$ in both sides of \eqref{eq:barv}. As a result, the following should hold:  
\be\begin{aligned}\label{eq:dvv2}
&|\!\ \!\underbrace{\Psi^\top (\mathrm{k}) B^\top \boldsymbol\theta(\mathrm{k})}_{a} +\underbrace{\tau \kappa\ \Psi^\top (\mathrm{k}) B^\top B |{\boldsymbol E}[{\underaccent{\sim}{\boldsymbol \alpha}}(\mathrm{k})]| \Psi(\mathrm{k})}_{b}\\&\hspace{-1mm}+\underbrace{\ \tau \Psi^{\top}(\mathrm{k}) {\boldsymbol E}[B^{\top} {\tilde{\boldsymbol \omega}}(\mathrm{k})]}_{c}| \geq \Psi^\top (\mathrm{k}) B^\top \boldsymbol\theta(\mathrm{k}).
\end{aligned}\ee
Since $\Psi(\cdot)$ is an odd function, $b>0$ holds (as discussed above). Then, if $a+b\pm |c|>0$, to have a negative drift, $\min \{b\} > \max \{|c|\}$ should hold, which gives the condition on $\kappa$ as in \eqref{eq:kappa-12}. If the latter condition holds, $a+b \pm |c|$ is always positive. However, the relative phase at each time should be smaller than $\pi$ (definition of geodesic distance), thus $\max\{b+c\} < \Psi_{\max} (\pi-\gamma_{\max})$ should hold, which completes the proof.
\end{proof}
%---------------------------------------------------------------------------------------------------------------------------------------
\section{Proof of Corollary~\ref{th2}}\label{ath2}
\begin{proof}
The proof is based on Theorem 11.0.1 of \cite{meyn2012markov}. We derive conditions under which the following inequality holds
\begin{equation}\label{eq:drift2}
\boldsymbol E[V(\Phi_{k+1})|\Phi_k=x]-V(x)< -1, \quad \;\forall \Phi_k \in \Pi \setminus S^{\Psi}_{G}(\gamma).
\end{equation}
The rest of the proof is similar to the proof of Theorem~\ref{th1}.
\end{proof}
%---------------------------------------------------------------------------------------
\section{Proof of Proposition\ref{co:clus}}\label{aclus}
\begin{proof}
Take ${\bar V}={\underaccent{i,j}{\max}} |\Psi(\theta_{i,j})|$ and assume that at time $\mathrm{k}$, $\{\exists (\ell,e) \in \E: |\theta_{\ell,e}| \in \Upsilon, {\bar V}(\mathrm{k})=|\Psi(\theta_{\ell,e})|\}$.  
To prove the stochastic phase-cohesiveness w.r.t. $U^{\Psi}_{G}(\gamma)$, we use a drift-based argument to derive conditions under which $\Delta {\bar V}=\boldsymbol E[{\bar V}(\mathrm{k}+1)]-|\Psi(\theta_{\ell,e})(\mathrm{k})|$ is negative. Recall that from the definition of the desired sets in Section~\ref{sec:set}, $\Psi(\cdot)$ for all arcs that belong to $\Upsilon$ is larger than arcs in $\underline{\Upsilon}$ or $\overline{\Upsilon}$. So, instead of computing $\Psi(\theta_{\ell,e} (\mathrm{k}+1))$, we use an equivalent argument. We compute the evolution of edge (relative phase) $\theta_{\ell,e}$, that is $\Delta \theta_{\ell,e}={\boldsymbol E}[|\theta_{\ell,e}(\mathrm{k}+1)|]-|\theta_{\ell,e}(\mathrm{k})|$. We prove that an edge for which $\sign({\boldsymbol E}[{\tilde\alpha}_{\ell,e}])>0$ holds is recurrent to $\underline{\Upsilon}$ while if $\sign({\boldsymbol E}[{\tilde\alpha}_{\ell,e}])<0$ holds, it is recurrent to $\overline{\Upsilon}$. In both cases, $\Delta {\bar V} <0$ is guaranteed.   
Assume $\theta_{\ell,e}(\mathrm{k}) >0$. First consider the case ${\boldsymbol E}[|\theta_{\ell,e}(\mathrm{k}+1)|]={\boldsymbol E}[\theta_{\ell,e}(\mathrm{k}+1)]$. Since the underlying topology is a line graph, based on~\eqref{eq:p11}, we have 
\be\begin{aligned}\label{eq:mixm}
\Delta &\theta_{\ell,e}\!=\!{\boldsymbol E}[{\theta_{\ell,e}}(\mathrm{k}+1)]-{\theta_{\ell,e}}(\mathrm{k})=
\\&- 2 \ \tau\  \kappa\ \sign({\boldsymbol E}[{\tilde\alpha}_{\ell,e}(\mathrm{k})]) |{\boldsymbol E}[{\tilde\alpha}_{\ell,e}(\mathrm{k})]| \Psi(\theta_{\ell,e}(\mathrm{k}))
\\&- \tau\ \kappa\ \sign({\boldsymbol E}[{\tilde\alpha}_{p,q}(\mathrm{k})]) |{\boldsymbol E}[{\tilde\alpha}_{p,q}(\mathrm{k})]| \Psi(\theta_{p,q}(\mathrm{k}))
\\&- \tau\ \kappa\ \sign({\boldsymbol E}[{\tilde\alpha}_{d,s}(\mathrm{k})]) |{\boldsymbol E}[{\tilde\alpha}_{d,s}(\mathrm{k})]| \Psi(\theta_{d,s}(\mathrm{k})),\end{aligned}\ee
where $(p,q)$ and $(d,s)$ denote the neighboring edges. Our aim is to prove that $\sign({\boldsymbol E}[{\tilde\alpha}_{\ell,e}(\mathrm{k})]) \Delta \theta_{\ell,e} <0$. Recall that $|\Psi(\theta_{\ell,e}(\mathrm{k}))| \geq \Psi(\gamma)$, where $|\Psi(\theta_{\ell,e}(\mathrm{k}))|$ is the maximum. That is, $|\Psi(\theta_{d,s}(\mathrm{k}))| \leq |\Psi(\theta_{\ell,e}(\mathrm{k}))|$ and $|\Psi(\theta_{p,q}(\mathrm{k}))| \leq |\Psi(\theta_{\ell,e}(\mathrm{k}))|$. Since the size of all mean-values are equal to $\lambda$ and $\Psi(\theta_{\ell,e}(\mathrm{k}))$ is maximum, the sign of $\Delta \theta_{\ell,e}$ is always equal to $-\sign({\boldsymbol E}[{\tilde\alpha}_{\ell,e}(\mathrm{k})])$ unless $|\Delta \theta_{\ell,e}|=0$. In fact, depending on the sign and size of $\Psi(\theta_{p,q}(\mathrm{k}))$ and $\Psi(\theta_{d,s}(\mathrm{k}))$, we have
$$0 \leq |\Delta \theta_{\ell,e}| \leq 4 \tau \kappa \lambda \Psi_{\max}.$$
Let us look at the case of $|\Delta \theta_{\ell,e}|=0$. In this case, the $\Psi(\cdot)$ of the neighboring edges are equal to $\Psi(\theta_{\ell,e}(\mathrm{k}))$. We consider the dynamics of either of them and write their evolution similar to \eqref{eq:mixm}. Now, either the same situation occurs or the maximum relative phase is dominant w.r.t. its neighbors. In case of a zero difference, we continue with a neighboring edge and repeat this process till reaching the tail of the graph or if we find a maximum edge whose neighbors (at least one neighbor) do not possess a corresponding maximum $\Psi(\cdot)$. Since the graph is a line graph, the final edge of the graph has only one neighbor. Then we can write, 
\be\begin{aligned}
\Delta &\theta_{\ell,e}\!=\!{\boldsymbol E}[{\theta_{\ell,e}}(\mathrm{k}+1)]-{\theta_{\ell,e}}(\mathrm{k})=
\\&- 2 \ \tau\  \kappa\ \sign({\boldsymbol E}[{\tilde\alpha}_{\ell,e}(\mathrm{k})]) |{\boldsymbol E}[{\tilde\alpha}_{\ell,e}(\mathrm{k})]| \Psi(\theta_{\ell,e}(\mathrm{k}))
\\&- \tau\ \kappa\ \sign({\boldsymbol E}[{\tilde\alpha}_{p,q}(\mathrm{k})]) |{\boldsymbol E}[{\tilde\alpha}_{p,q}(\mathrm{k})]| \Psi(\theta_{p,q}(\mathrm{k})).\end{aligned}\ee
From the above we conclude that for $\kappa >0$, $|\Delta \theta_{\ell,e}| \geq \tau \kappa \lambda \Psi_{\gamma} \neq 0$ and $\sign(\Delta \theta_{\ell,e})= - \sign({\boldsymbol E}[{\tilde\alpha}_{\ell,e}])$. This gives $\Delta {\bar V} <0$. Now, consider the case of ${\boldsymbol E}[|\theta_{\ell,e}(\mathrm{k}+1)|]= -{\boldsymbol E}[\theta_{\ell,e}(\mathrm{k}+1)]$. Writing the inequality in \eqref{eq:mixm} for the maximum case, we obtain the condition on the sampling time, \ie, $-\gamma +4 \tau \kappa \lambda \Psi_{\max} < \gamma$ which completes the proof. 
\end{proof}
%------------------------------------------------------------------------------
\section{Proof of Proposition~\ref{co:sh}}\label{ash}
\begin{proof}
Similar to the proof of Theorem~\ref{th1}, define
\begin{equation*}\begin{aligned}
V(\Theta({\mathrm{k}}))&=|\Psi_{r}(\gamma)| \sum\limits_{|\theta_{i,j}(\mathrm{k})| \in {\Upsilon}} \mu_{i,j}\; |\theta_{i,j}(\mathrm{k})| + \\ 
&\Psi^{o} \sum\limits_{|\theta_{i,j}(\mathrm{k})| \in \underline{\Upsilon}^{1} \cup \underline{\Upsilon}^{2}} \mu_{i,j}\; |\theta_{i,j}(\mathrm{k})|
\end{aligned}\end{equation*}
where $\mu_{i,j}={\boldsymbol E}[{\tilde\alpha}_{i,j}]$. Assume that at time $\mathrm{k}$, there exists an edge $\theta_{\ell,e}$ such that $|\theta_{\ell,e}| \in \Upsilon$. To characterize the lower bound on $\kappa$, we assume the worst condition, \ie,  $|(\theta_{\ell,q} (\mathrm{k}))|= \gamma$ and $|\theta_{p,s} (\mathrm{k})| \leq \gamma_c, \forall (p,s) \neq (\ell,q)$ and $\Psi(\theta_{p,s} (\mathrm{k}))=-\bar\Psi$. Similar to the proof of Theorem~\ref{th1}, it holds $V(\Theta({\mathrm{k}+1})) \leq {\bar V}(\Theta({\mathrm{k}+1}))$, where 
$${\bar V}(\Theta({\mathrm{k}+1}))=|\Psi_{r,\mu}^\top (B^\top \boldsymbol\theta(\mathrm{k}))| |B^\top \boldsymbol\theta(\mathrm{k}+1)|,$$
with $\Psi_{r,\mu}(B^\top \boldsymbol\theta(\mathrm{k}))= \mu \Psi_{r}(B^\top \boldsymbol\theta(\mathrm{k}))$, and $\mu_{m \times m}$ is a constant matrix equal to $\mu= {\boldsymbol E}[{\underaccent{\sim}{\boldsymbol \alpha}}]$. Different from Theorem~\ref{th1}, there exists an arc in $[0,\pi]$ on which $\Psi_{r}(\cdot)$ may take positive or negative values. Hence, $\Psi_{r,\mu}^\top (B^\top \boldsymbol\theta(\mathrm{k})) B^\top \boldsymbol\theta(\mathrm{k}) \geq 0$ does not necessarily hold. Denote ${\bar V}(\Theta({\mathrm{k}}))$, $\Psi_{r,\mu}^\top (B^\top \boldsymbol\theta(\mathrm{k}))$ and $\Psi_{r} (B^\top \boldsymbol\theta(\mathrm{k}))$ by ${\bar V}({\mathrm{k}})$, $\Psi_{r,\mu}^\top (\mathrm{k})$ and $\Psi_{r}(\mathrm{k})$, respectively. Since, the variances are assumed small, we have
\begin{equation}\begin{aligned}\label{eq:coin}
&\hspace{-40mm}\Delta {\bar V} (\Theta) \leq - {\bar V}(\mathrm{k})+{\boldsymbol E} [\underbrace{\ \tau |\Psi_{r,\mu}^{\top}(\mathrm{k})|\;|B^{\top} {\tilde{\boldsymbol \omega}}(\mathrm{k})}_{c}| \;\big|{\boldsymbol\theta}(\mathrm{k})]+ \hspace{-10mm} \\ \!\!\ \! {\boldsymbol E} [\underbrace{|\Psi_{r,\mu}^\top (\mathrm{k})| B^\top \boldsymbol\theta(\mathrm{k})}_{a} -\tau \kappa\ &\underbrace{|\Psi_{r,\mu}^\top (\mathrm{k})| B^\top B {\underaccent{\sim}{\boldsymbol \alpha}}(\mathrm{k}) \Psi_{r}(\mathrm{k})}_{b} \;\big|{\boldsymbol\theta}(\mathrm{k})].
\end{aligned}\end{equation}
Similar to the proof of Theorem \ref{th1}, to have $\Delta {\bar V} <0$, the inequalities in \eqref{eq:2con} should hold. Now, we continue by  characterizing the lower bound on $\kappa$, we assume the worst condition, \ie,  $|(\theta_{\ell,q} (\mathrm{k}))|= \gamma$ and $|\theta_{p,s}(\mathrm{k})| < \gamma_c, \forall (p,s) \neq (\ell,q)$ and $\Psi(\theta_{p,s} (\mathrm{k}))=-\bar\Psi$. Considering the worst condition, assume that
${\boldsymbol E}[{\tilde\alpha}_{\ell,q}]= \lambda_{\min}({\boldsymbol E}[{\underaccent{\sim}{\boldsymbol \alpha}}])$ and ${\boldsymbol E}[{\tilde\alpha}_{p,s}]= \lambda_{\max}({\boldsymbol E}[{\underaccent{\sim}{\boldsymbol \alpha}}])$. From the definition of ${\bar V}(\boldsymbol\theta({\mathrm{k}}))$, we obtain $a={\bar V}(\mathrm k)$. Define, 
$$x=\lambda_{\min}({\boldsymbol E}[{\underaccent{\sim}{\boldsymbol \alpha}}])[|\Psi_{r}(\gamma)|\  0 \ldots 0]^\top,$$ 
$$y=\lambda_{\max}({\boldsymbol E}[{\underaccent{\sim}{\boldsymbol \alpha}}])[0 \ \bar\Psi\ \bar\Psi \ldots \bar\Psi]^\top.$$
Then, we can write $|\Psi_{r,\mu}(\mathrm{k})|= x+y$ and ${\boldsymbol E}[{\underaccent{\sim}{\boldsymbol \alpha}}(\mathrm{k})] \Psi_{r}(\mathrm{k})= x-y$. Thus, computing the term $b$ in \eqref{eq:coin} gives
$${\boldsymbol E}[b] = (x^\top+y^\top) B^\top B (x-y)= x^\top B^\top B x- y^\top B^\top B y.$$ 
Notice that $B^\top B$ is symmetric, and for a tree graph $L_e= B^\top B$ is positive definite. Define ${\hat\lambda}_1=\lambda_{\min}({\boldsymbol E}[{\underaccent{\sim}{\boldsymbol \alpha}}]) |\Psi_{r}(\gamma)|$ and ${\hat\lambda}_2=\lambda_{\max}({\boldsymbol E}[{\underaccent{\sim}{\boldsymbol \alpha}}]) \bar{\Psi}$. Assume that 
$${\hat\lambda}={\hat\lambda}_1-{\hat\lambda}_2 \sqrt{(m-1)}> 0.$$
Denote $f={\hat\lambda}_1+{\hat\lambda}_2 \sqrt{(m-1)}>0$. As a result, ${\boldsymbol E}[b] \geq \lambda_{\min}(L_{e}) {\hat\lambda} f >0.$ To obtain $\kappa$, we should have $\kappa \min\ \{{\boldsymbol E}[b]\} > \max\ {\boldsymbol E}[c]$. This gives,
$$\kappa > \frac{\big({\hat\lambda}_1+{\hat\lambda}_2 (m-1)\big) {\boldsymbol E}_{\max} [|\Delta{\tilde{\omega}}|]}{\lambda_{\min}(L_{e}) {\hat\lambda} \big({\hat\lambda}_1+{\hat\lambda}_2 \sqrt{(m-1)}\big)}.$$
To obtain the bound on $\tau$, we assume that all edges of the network belong to $\Upsilon$. The rest of the proof follows from the proof of Theorem~\ref{th1} which leads to the conditions in \eqref{eq:kappa-sh}.
\end{proof}
%--------------------------------------------------------------------------------
\section{Proof of Lemma~\ref{pr3}}\label{apr3}
\begin{proof}
The network topology is a random graph, therefore, at each time step $\mathrm{k}$, the set of randomly established edges determines the space wherein the relative phases evolve until $\mathrm{k+1}$. Thus, the probability space is the countable set of spaces that are randomly selected by multiple independent Bernoulli processes. The probability space %is countable because it 
may contain at least one space (correspond to the null graph) and at most $2^m$ spaces (corresponding to the maximal graph). Since $p$ is a fixed non-zero probability and the topology of the network at a time $\mathrm{k}$ is independent of the states $\theta_i(\mathrm{t})$, $\mathrm{t}<\mathrm{k}$, the dynamics $\theta_i(\mathrm{k+1})$ only depends on $\theta_i(\mathrm{k})$, and the random topology determined by $p$ and hence \eqref{eq:rand1} is a Markov chain. Due to the independence of the Bernoulli processes over time and also $p$ being a non-zero probability, the countable set of probability spaces is $\psi$-irreducible (Definition \ref{def:irreducibility}, also \cite[Ch.4]{meyn2012markov}). 
\end{proof}
%--------------------------------------------------------------------------------------
\section{Proof of Theorem~\ref{th3}}\label{ath3}
\begin{proof}
The proof is based on Theorem 11.0.1 of \cite{meyn2012markov} for chains evolving in a countable space. Notice that here, all compact sets are petite since the space is countable. The trend of the proof is similar to the proof of Theorem \ref{th1}. Define function $V(\Theta({\mathrm{k}}))$ as in \eqref{eq:lya}. Different from Theorem \ref{th1}, here, the randomness are governed by Bernoulli distribution and we can compute the finite set of all possible outcomes. In order to compute the one-step drift, for $V_m$ and $V_M$ (as defined in the proof of Theorem \ref{th1}), we need to calculate
${\boldsymbol E}[|I_1|]$, with $I_1=B^\top \boldsymbol\theta(\mathrm{k})-\tau \kappa B^\top B \beta(\mathrm{k}) \Psi(B^\top \boldsymbol\theta(\mathrm{k})).$ Since each two oscillators, if connected, are linked un-directedly, we can write $B$ such that $B^\top \boldsymbol\theta({\mathrm{k}})\succ 0$ holds. Recall that based on the definition of $V_m$, only one edge is non-zero, and for $V_M$, all are non-zero. Now, consider $I_1$, also in the view of dynamics in \eqref{eq:relran}. If $\tau$ is sufficiently small, for a given $\kappa$, we can assure that all elements of $I_1$ are positive, hence, ${\boldsymbol E}[|I_1|]={\boldsymbol E}[I_1]$ holds. To characterize the condition on $\tau$, we write
$$\gamma=\min{|\theta_{i,j}(\mathrm{k})|} \geq \tau \kappa \ d_{\max} \Psi_{\max},$$%+|\Delta_{\max} \omega|\big),$$
where $d_{\max}$ is the maximum degree, and $d_{\max} \leq \lambda_{\max}(L_e)$. This gives the bound on $\tau$. Then, based on a similar argument as the proof of Theorem \ref{th1}, we show that the one-step drift of $V$ from ${\Upsilon}$ in~\eqref{eq:a2} to $\underline{\Upsilon}$ in~\eqref{eq:a1} is negative if $\kappa$ is sufficiently large, as in \eqref{eq:kappa-3}. Furthermore, since $\forall i,j, \omega_{i,j} \neq 0$ and the oscillators are randomly connected, if the maximum relative phase enters $\overline{\Upsilon}$, the probability that it exits this set is non-zero. In a similar fashion to Example \ref{exam1} and Theorem \ref{th1}, we argue that $S^{\Psi}_{G}(\gamma_{\max})$ is transient which ends the proof.
\end{proof}
%--------------------------------------------------------------------------------------
\section{Proof of Corollary~\ref{th4}}\label{ath4}
\begin{proof}
The proof follows a similar trend as of Theorem \ref{th3}. Set $\gamma=0^+$ which gives $S^{\Psi}_{G}(0)=\{\theta_i, \theta_j \in \S^1:|\theta_{i,j}(\mathrm{k})|=\gamma, \forall (i,j) \in \mathcal{E}\}$. 
From the result of Theorem~\ref{th3}, we have 
\be\label{eq:last} \kappa\ \tau\ |\Psi(\gamma)|\ p\ \lambda_{\min}(L_e) > |\Delta_{\max}{\omega}|.\ee
Notice that the above holds for any $\gamma \in [0^+,\pi^-]$. By substituting $|\Delta_{\max}{\omega}|=0$, we conclude that the condition is satisfied $\forall \kappa >0$. From the proof of Theorem~\ref{th3}, we have $\tau \sim \frac{\gamma}{\kappa}$, and hence sufficiently small. This indicates that if the chain exits the origin, it revisits the origin with probability one. Notice that this does not hold for the anti-phase arc $\overline{\Upsilon}={\pi}$. Since, exiting this arc, the relative phases will return to the origin as explained above. This completes the proof.
\end{proof} 
%================================================================================================================================
\bibliographystyle{ieeetr}
\bibliography{biblio}

\end{document}